\documentclass{amsart}

\usepackage{amsmath,amssymb,amsfonts,amsthm}
\usepackage[all]{xy}
\usepackage{color}
\usepackage{hyperref}
\usepackage{multirow}
\usepackage{rotating}

\usepackage{stmaryrd}

\renewcommand{\(}{\begin{equation}}
\renewcommand{\)}{\end{equation}}
\newcommand{\bea}{\begin{eqnarray}}
\newcommand{\eea}{\end{eqnarray}}
\newcommand{\R}{{\mathbb R}}

\newcommand{\LX}{\bigwedge^\bullet\mathfrak{X}}

\newcommand{\Sh}{\mathrm{Sh}}
\newcommand{\maps}{\colon}
\newcommand{\tensor}{\otimes}
\newcommand{\vs}[1]{-(-1)^{\binom{{#1}+1}{2}}}
\renewcommand{\deg}[1]{\left \lvert #1 \right \rvert}

\newcommand{\vk}[1]{v_{1} \wedge \cdots \wedge  v_{#1}}
\newcommand{\vsk}[1]{v_{\sigma(1)} \wedge \cdots \wedge v_{\sigma(#1)}}
\newcommand{\xto}[1]{\xrightarrow{#1}}

\DeclareMathOperator{\res}{\mathrm{res}}
\DeclareMathOperator{\id}{\mathrm{id}}

\DeclareMathOperator{\Tot}{\mathrm{Tot}}
\newcommand{\LieQuant}{\mathrm{dgLie}_{\mathrm{Qu}} (X,\bar{A})}

\newcommand{\Xham}{\mathfrak{X}_{\mathrm{Ham}}}
\newcommand{\X}{\mathfrak{X}}
\renewcommand{\L}{\mathcal{L}}
\newcommand{\bbrac}[2]{\left \llbracket #1,#2 \right \rrbracket}

\newcommand{\bx}{\bar{x}}
\newcommand{\btheta}{\bar{\theta}}
\newcommand{\ksgn}[1]{-(-1)^{\binom{#1+1}{2}}}

\newcommand{\fc}{\mathfrak{c}}
\newcommand{\fa}{\mathfrak{a}}

\newcommand{\dt}{d_{\mathrm{Tot}}}
\newcommand{\cU}{\mathcal{U}}

\numberwithin{equation}{subsection}

\renewcommand{\deg}[1]{\left \lvert #1 \right \rvert}

\renewcommand{\L}{\mathcal{L}}
\newcommand{\innerprod}[2]{\langle #1,#2 \rangle}

\newcommand{\nbbt}[1]{\textcolor{black}{#1}}

\def\gg {\mathfrak{g}}

\theoremstyle{plain}
\newtheorem{theorem}{Theorem}[subsection]
\newtheorem{lemma}[theorem]{Lemma}

\newtheorem{proposition}[theorem]{Proposition}
\newtheorem{corollary}[theorem]{Corollary}
\newtheorem{propositionapp}{Proposition}[section]
\newtheorem{lemmaapp}[propositionapp]{Lemma}

\theoremstyle{definition}
\newtheorem{definition}[theorem]{Definition}
\newtheorem{defprop}[theorem]{Definition/Proposition}
\newtheorem{example}[theorem]{Example}

\newtheorem{remark}[theorem]{Remark}

\begin{document}

\title{$L_\infty$-algebras of local observables from higher prequantum bundles
}
\author{Domenico Fiorenza}
\email{fiorenza@mat.uniroma1.it}
\address{Department of Mathematics, Sapienza Universit\`a di Roma,
P.le Aldo Moro 2, 00185, Rome, Italy.}

\author{Christopher L. Rogers}
\email{crogers@uni-math.gwdg.de}
\address{Mathematics Institute, Georg-August Universit\"at G\"ottingen,
Bunsenstrasse 3-5, D-37073,
G\"ottingen, Germany}

\author{Urs Schreiber}
\email{urs.schreiber@gmail.com}
\address{Mathematics Institute,
Radboud Universiteit Nijmegen,
Comeniuslaan 4, 6525 HP,
Nijmegen, The Netherlands}

\begin{abstract}

To any manifold equipped with a higher degree closed form, one can associate
an $L_\infty$-algebra of local observables that generalizes the
Poisson algebra of a symplectic manifold. Here, by means of an
explicit homotopy equivalence, we interpret this $L_\infty$-algebra in
terms of infinitesimal autoequivalences of higher prequantum
bundles. By truncating the connection data on the prequantum
bundle, we produce analogues of the (higher) Lie algebras of
sections of the Atiyah Lie algebroid and of the Courant Lie
2-algebroid. We also exhibit the $L_\infty$-cocycle that realizes
the $L_\infty$-algebra of local observables as a Kirillov-Kostant-Souriau-type
$L_\infty$-extension of the Hamiltonian vector fields. When restricted
along a Lie algebra action, this yields Heisenberg-like
$L_\infty$-algebras such as the string Lie 2-algebra of a semisimple
Lie algebra.
\end{abstract}

\subjclass[2010]{53D50; 53C08; 18G55}
\keywords{Geometric quantization, gerbes, homotopical algebra}

\maketitle

\setcounter{tocdepth}{1} 
\tableofcontents

\section{Introduction}

Geometric objects, such as
manifolds, orbifolds, or stacks, equipped with a closed differential
form play important roles in many areas of current mathematical
interest.  The archetypical examples are closed 2-forms in
(pre-)symplectic geometry. Higher
degree closed forms play crucial roles, for example, in covariant quantum field
theory, in Hitchin's generalized complex/Riemannian geometry, and in differential
cohomology. It is becoming clear that it is advantageous to
consider these forms, in one way or another, as higher degree generalizations of symplectic
structures.

In all of these applications, there is a particular focus on integral
closed forms. This is because such forms correspond to the curvatures
of higher geometric structures known as $U(1)$-$n$-bundles with
connection (or $U(1)$-$(n-1)$-bundle gerbes with connection).  Here we
refer to these as \emph{higher prequantum bundles}, in analogy with
the role that $U(1)$-principal connections play in the geometric
prequantization of symplectic manifolds {\cite{Kostant:1970,Souriau:1967}}.
(A modern review can be found in \cite{Brylinski:1993}.)  In the companion article
\cite{companionArticle} we present general aspects of such higher
geometric prequantum structures; here we work out details of the
general theory specialized to the higher differential geometry over
smooth manifolds.  In particular,
we use homotopy Lie theory to study the infinitesimal autoequivalences of
higher prequantum bundles covering infinitesimal diffeomorphisms of
the base manifold, i.e., the infinitesimal quantomorphisms.

It is well known
that every pre-symplectic manifold induces a Lie algebra
of Hamiltonian functions whose bracket is the Poisson bracket given
by the closed 2-form. When the manifold is equipped with a prequantum
bundle, this Lie algebra is isomorphic to the Lie algebra of
infinitesimal autoequivalences of that structure, i.e., those vector fields on the
bundle whose flow preserves the underlying bundle and its connection under pullback.
These are also called prequantum operators.
More generally, manifolds equipped with higher degree
forms also have Hamiltonian vector fields, which form a Lie algebra just as in
symplectic geometry. The differential form induces a bilinear
skew-symmetric bracket not on functions, but on higher degree
differential forms. However, this bracket fails to satisfy the Jacobi identity.
The observation made in \cite{Rogers:2010nw} was that, for the case of
non-degenerate forms, this failure  is controlled by coherent homotopy.
Hence, instead of being a problem, the lack of a genuine Lie bracket
indicates the presence of a natural, but higher (homotopy-theoretic)
structure. More precisely,  the higher Poisson bracket gives rise to
a strong-homotopy Lie algebra or $L_\infty$-algebra.
The construction in \cite{Rogers:2010nw} extends
immediately
to the case of degenerate forms, and we call the resulting algebra
the `$L_\infty$-algebra of local observables'. In this paper,
we illuminate its conceptual role further.

\medskip

\noindent {\bf Summary of results.}
We identify the
higher Kirillov-Kostant-Souriau $L_\infty$-algebra cocycle
that classifies the 
$L_\infty$-algebra of local observables
as an extension of the Hamiltonian vector fields (theorem \ref{TheExtensionClassifiedByTheCocycle})
and
show how this result immediately
gives a construction
of `higher Heisenberg $L_\infty$-algebras'
(section \ref{TheHeisenbergExtension}). As an example, we obtain
a direct rederivation (example \ref{StringAsHeisenberg})
of the $\mathfrak{string}_{\mathfrak{g}}$-Lie 2-algebra
as the Heisenberg Lie 2-algebra of a compact simple Lie
group $G$ \cite{RogersString}.

We briefly recall the construction of the higher prequantum automorphism
group of a higher prequantum bundle, which is described with more detail
in \cite{companionArticle}. {We construct a dg Lie algebra (def.\
  \ref{dgLieXOmega}) that can be thought of as modeling the
  ``infinitesimal elements'' of this higher automorphism group in
  terms of the {\v C}ech-Deligne cocycle for the prequantum bundle.
(Similar dg Lie models for the ``infinitesimal symmetries''
of a $U(1)$-bundle gerbe were constructed by Collier \cite{Collier:2011}.)}

{We prove explicitly that our dg Lie algebra of infinitesimal quantomorphisms is
equivalent, as an $L_\infty$-algebra, to the
$L_\infty$-algebra  of local observables of the corresponding
pre-$n$-plectic form (theorem \ref{theorem.main-theorem}).}

Finally,
we show that this construction induces an inclusion of
the $L_\infty$-algebra of local observables
into higher Courant and higher Atiyah Lie algebras
(section \ref{InclusionIntoAtiyahAndCourant}).

\begin{remark}
  All of the constructions and results that we discuss here apply to
  the general context of pre-$n$-plectic manifolds, i.e., manifolds
  equipped with a closed $(n+1)$-form. Non-degeneracy conditions on
  the differential form do not play a role. Nevertheless, our
  formalism allows us to restrict to the case of non-degenerate forms,
  and it may be interesting to do so in specific applications.  This
  is analogous to the well-known fact that non-degeneracy is not
  needed to prequantize a symplectic manifold. Indeed, one can proceed
  even further in this case; the full geometric quantization of
  pre-symplectic manifolds is a well-defined and interesting endeavour
  in its own right (e.g.\ \cite{SKT}).

\end{remark}
\medskip

\medskip

\noindent {\bf Motivation and perspective.} The $L_\infty$-algebras of local observables as considered here
appear naturally in traditional field theory in the guise of higher
order local Noether currents. For instance,
it is shown in \cite{Baez:2008bu} how the energy-momentum tensor for the
bosonic string arises in the Lie 2-algebra associated to a multiphase
space for a 1+1 dimensional field theory.
Generally, the classical
Hamilton-de Donder-Weyl field equations in multisymplectic field theory
characterize the higher dimensional infinitesimal flows
in the $L_\infty$-algebra of local observables (Maurer-Cartan elements in the tensor product
with a Grassmann algebra); this is discussed in section
1.2.11.3  of \cite{dcct}.

In a broader perspective, these $L_\infty$-algebras naturally arise in the context of higher geometric prequantization and in particular in the geometric quantization of loop groups by the orbit method, see, e.g., \cite[p.\ 249]{Brylinski:1993} and the discussion in \cite[Sec.\ 2.6.1]{companionArticle}.
This  was
a motivation behind the
  refinement of multisymplectic geometry to homotopy theory developed
  in \cite{Rogers:2010nw}, leading  to a higher Bohr-Sommerfeld-like
  geometric quantization procedure for manifolds equipped with closed
  integral 3-forms \cite[Chap.\ 7]{RogersThesis}.
  These integral 2-plectic stuctures also naturally appear as
  the geometric quantization of Poisson manifolds
  via their associated symplectic groupoids
  (whose multiplicative symplectic form is secretly a 2-plectic simplicial form), see \cite{Bongers13}.

In terms of quantum field theory, higher geometric prequantization
concerns the pre-quantum incarnation of {local} quantum field theory, in the way envisioned by
Freed \cite{Freed}, Baez-Dolan \cite{Baez-Dolan}, and more recently formalized by Lurie \cite{Lurie}.
While Lurie's theorem gives a full characterization of the
topological quantum field theories that are local in this sense,
it is an open problem to find a refinement of the process of quantization
that would ``read in'' higher geometric prequantum data and produce a local
QFT in this sense. The results of the present article, when placed within
the larger context of higher pre-quantum geometry, as discussed more fully in \cite{companionArticle},
are meant to provide some answers to this open question.
Indeed, based on these developments, further progress in this
direction has been made recently in \cite{Nuiten}. A survey is given in section 6 of \cite{dcct}.

It should be remarked that in the present article we are solely interested in the $L_\infty$-algebra structure on local observables and we are not investigating the existence of compatible associative and commutative algebra structures (up to homotopy) making the higher local observables a $\text{Poisson}_\infty$-algebra. This issue will hopefully be investigated elsewhere. It is also worth mentioning that, in parallel to the $L_\infty$-algebras for $n$-plectic geometry as considered here,
there are various other attempts to formulate generalizations of the algebraic structures found in symplectic geometry to multisymplectic geometry \cite{Forger,Kanatch, Richter}. These
differing proposals are not manifestly equivalent, and it would be interesting to understand the relations between these various proposals at a deeper level.

\medskip

\noindent{\it Acknowledgements.} We thank the referee, Johannes Huebschmann, and Jim Stasheff for valuable comments and suggestions on a first version of this article, Yael Fr\'{e}gier and Marco
Zambon for sharing a preliminary draft \cite{FRZ:2013} of their work (joint with
C.L.R.) that inspired some of the ideas developed here, and Ruggero
Bandiera, Christian Blohmann, and Marco Manetti for the invaluably inspiring conversations they
had with D.F.\ on the homotopy fibers of $L_\infty$-morphisms and on
higher symplectic geometry.

C.L.R.\ acknowledges support from the German Research Foundation
(Deutsche Forschungsgemeinschaft (DFG)) through the Institutional
Strategy of the University of G\"{o}ttingen. U.S. was supported by
the Dutch Research Organization (NWO project 613.000.802).

\subsection*{Notation and conventions.}
\subsubsection{Notation for Cartan calculus} \label{notat_cartan}
The Schouten bracket of two decomposable multivector fields
$u_{1} \wedge \cdots \wedge u_{m}, v_{1} \wedge \cdots \wedge v_{n}
\in \LX(X)$ is
\begin{multline} \label{schouten}
\left [ u_{1} \wedge \cdots \wedge u_{m}, v_{1} \wedge \cdots \wedge
  v_{n} \right]
= \\\sum_{i=1}^{m} \sum_{j=1}^{n} (-1)^{i+j} [u_{i},v_{j}]
\wedge u_{1} \wedge \cdots \wedge \hat{u}_{i} \wedge  \cdots \wedge
u_{m}
 \wedge v_{1} \wedge \cdots \wedge \hat{v}_{j} \wedge \cdots \wedge v_{n},
\end{multline}
where $[u_{i},v_{j}]$ is the usual Lie bracket of vector fields.

Given a form $\alpha \in \Omega^{\bullet}(X)$, the \textbf{interior product} of a decomposable
multivector field $v_{1} \wedge \cdots \wedge v_{n}$ with $\alpha$ is
defined as
$\iota_{v_{1} \wedge \cdots \wedge v_{n}} \alpha = \iota_{v_{n}} \cdots
\iota_{v_{1}} \alpha,
$ 
where $\iota_{v_{i}} \alpha$ is the usual interior product of vector
fields and differential forms. The interior product of an arbitrary
multivector field is obtained by extending the above formula by $C^\infty(X;\mathbb{R})$-linearity.
The \textbf{Lie derivative} $\L_{v}$ of a differential form along a multivector field $v \in
\LX(X)$ is defined via the graded commutator of $d$ and $\iota(v)$:
$
\L_{v} \alpha =  d \iota_v \alpha - (-1)^{\deg{v}} \iota_v d\alpha,
$ 
where $\iota(v)$ is considered as a degree $-\deg{v}$ operator.

The last identity we will need involving multivector fields is for the graded commutator of
the Lie derivative and the interior product. Given $u,v \in
\LX(X)$, we have the Cartan identity
\begin{equation} \label{Cartan_commutator}
\iota_{[u,v]} \alpha = (-1)^{(\deg{u}-1)\deg{v}} \L_{u} \iota_v  \alpha - \iota_v\L_{u} \alpha.
\end{equation}

\subsubsection{Conventions on chain and cochain complexes}
We will work mostly with chain complexes and homological degree conventions. The differential of a chain complex $(A_\bullet,d)$ will have degree $-1$:
$
\cdots\to A_{n+1}\xrightarrow{d} A_n \xrightarrow{d}A_{n-1}\to \cdots
$.
The shift functor $A_\bullet \mapsto A[1]_\bullet$ will act by $A[1]_k=A_{k-1}$. In particular, if $V$ is a vector space, seen as a chain complex concentrated in degree zero, $V[n]$ will be the chain complex consisting of $V$ concentrated in degree $n$.
A cochain complex $(A^\bullet,d)$ will have a differential of degree +1,
$
\cdots\to A^{n-1}\xrightarrow{d} A^n \xrightarrow{d}A^{n+1}\to \cdots,
$
 and will be identified with a chain complex (with the same differential) by the rule $A_k=A^{-k}$. In particular chain complexes concentrated in non-negative degree will correspond to cochain complexes concentrated in nonpositive degree, and vice versa. On cochain complexes the shift functor $A^\bullet \mapsto A[1]^\bullet$ will act by $A[1]^k=A^{k+1}$.

\subsubsection{Conventions and notation for $L_\infty$-algebras}
We will assume the reader is familiar with the homotopy theory of dg-Lie and $L_\infty$-algebras. A comprehensive account can be found in \cite{Loday-Vallette}.
We will follow homological degree
conventions, as in \cite{Lada-Markl:1995}, so that the differential of a dg-Lie algebra and of an $L_\infty$-algebra will have degree $-1$.
All examples of $L_\infty$-algebras $\mathfrak{g}$ given here will have
their underlying chain complex $\mathfrak{g}_\bullet$ concentrated in
non-negative degree. An $L_\infty$-algebra concentrated in degrees 0 through $(n-1)$
will be called a \emph{Lie $n$-algebra}.

An $L_\infty$-algebra whose $k$-ary brackets for $k \geq 2$
are trivial, i.e., a plain chain complex, is called an \emph{abelian} $L_\infty$-algebra.
If $\mathfrak{h}$ is an abelian $L_\infty$-algebra with underlying chain complex
$\mathfrak{h}_\bullet$, then we also write $\mathbf{B}\mathfrak{h}$ for the
abelian $L_\infty$-algebra with underlying chain complex $\mathfrak{h}_\bullet[1]$.
In particular, for $n \in \mathbb{N}$ we write $\mathbf{B}^n \mathbb{R}=\mathbb{R}[n]$ for the abelian $L_\infty$-algebra
whose underlying chain complex is $\mathbb{R}$ concentrated in degree $n$.

 An $L_\infty$-morphism of the form $\mathfrak{g} \to \mathbf{B}A$, for $A$ an abelian $L_\infty$-algebra, will be called an \emph{$L_\infty$-algebra cocycle} on the $L_{\infty}$-algebra
$\mathfrak{g}$ with coefficients in $A$. For $\mathfrak{g}$ a Lie algebra and $A=\mathbb{R}[n]$, these
  are just the traditional cocycles used in Lie
  algebra cohomology. See
\cite{NikolausSchreiberStevensonI,NikolausSchreiberStevensonII} for a discussion of $L_\infty$-algebra extensions in the broader context of principal $\infty$-bundles.

The (non-full) inclusion of dg-Lie algebras into $L_\infty$-algebras
is a part of an adjunction
  \begin{equation}  \label{TheAdjunction}
    (\mathcal{R} \dashv i)
	:
    \xymatrix{
	  L_\infty\mathrm{Alg}
	  \ar@<+4pt>@{<-}[r]^{i}
	  \ar@<-4pt>[r]_{\mathcal{R}}
	  &
	  \mathrm{dgLie}
	}
	\,,
  \end{equation}
  see for instance \cite[Proposition
  11.4.5]{Loday-Vallette}.
  We will call $i \circ \mathcal{R}$ the \emph{rectification functor}
  for $L_\infty$-algebras, and will often leave the (non-full)
  embedding $i$ notationally implicit.
In particular,
for any $L_\infty$-algebra $\mathfrak{g}$
there is a \emph{canonical} $L_\infty$-algebra homomorphism
$\mathfrak{g} \stackrel{v_{\mathfrak{g}}}{\to} \mathcal{R}(\mathfrak{g})$, namely,  the unit of the adjunction, such that
every $L_\infty$ morphism $f_\infty:\mathfrak{g}\to A$ to a dg-Lie algebra $A$
uniquely factors as
$
\mathfrak{g}\xrightarrow{v_{\mathfrak{g}}} \mathcal{R}(\mathfrak{g})\xrightarrow{\xi_A\circ \mathcal{R}(f_\infty)}A,
$
where $\xi_A:\mathcal{R}(A)\to A$ is the dg-Lie algebra morphism in the factorization of the identity of $A$ as
$
A\xrightarrow{v_{A}} \mathcal{R}(A)\xrightarrow{\xi_A}A$.

 There is a wealth of presentations for the homotopy theory of
 $L_\infty$-algebras,
 given by a web of model category structures
 with Quillen equivalences between them \cite{Pridham}.
 Here we make use of the model structures due to \cite{Hinich, Hinich:2001},
 from which one can distill the following statement: 
   the category 
    of dg-Lie algebras (over the real numbers)
   carries a model category structure in which the weak equivalences
   are the quasi-isomorphisms on the underlying chain complexes,
   and the fibrations are the degreewise surjections on the underlying chain
   complexes.
   Moreover, if we define a morphism $\mathfrak{g} \to \mathfrak{h}$
   in $L_\infty\mathrm{Alg}$ to be a weak equivalence iff the
   underlying morphism of complexes $\mathfrak{g}_\bullet \to
   \mathfrak{h}_\bullet$ is a quasi-isomorphism, then the adjunction
   $(\mathcal{R} \dashv i)$ 
   induces an equivalence between the homotopy theories of dg-Lie algebras
   and $L_\infty$-algebras. In particular, the components of the unit
   $\mathfrak{g} \stackrel{v_{\mathfrak{g}}}{\to}
   \mathcal{R}(\mathfrak{g})$ and counit $\mathcal{R}(A) \xto{\xi_{A}}
   A$ of this adjunction
   are weak equivalences.
   
\subsubsection{Conventions on stacks and higher stacks}
While this article focuses on homotopy Lie theory, we do mention
at some points the corresponding constructions in
higher smooth stacks, according to \cite{companionArticle}.
A detailed overview of this formalism is given in Sec.\ 3.1 in
\cite{FScSt}.
Smooth stacks are taken to be stacks over the category of all
smooth manifolds equipped with its standard Grothendieck topology
of good open covers. Equivalently but more conveniently these are
stacks over just the subcategory $\mathrm{CartSp}$
of Cartesian spaces $\{\mathbb{R}^n\}_{n \in \mathbb{N}}$ (or equivalently of open $n$-balls),
regarded as smooth manifolds. A higher smooth stack may always be presented
as a Kan-complex valued functor on $\mathrm{CartSp}^{\mathrm{op}}$ and the
homotopy theory $\mathbf{H}$ of smooth stacks is given by the category of such functors
with stalkwise homotopy equivalences of Kan complexes universally turned
into actual homotopy equivalences:
$
  \mathbf{H} := L_{\mathrm{lhe}} \,\mathrm{Func}(\mathrm{CartSp}^{\mathrm{op}}, \mathrm{KanCplx})
  \,.
$
In the applications of the present
article all examples of such objects are either given by sheaves of
chain complexes $A_\bullet$ of abelian groups in non-negative degrees under the Dold-Kan correspondence
$
  \mathrm{DK}
  :
    \mathrm{Ch}_{\geq \bullet}(\mathrm{Ab})
    \xrightarrow{\simeq}
    \mathrm{AbGrp}^{\Delta^{\mathrm{op}}}
    \xrightarrow{\mathrm{forget}}
    \mathrm{KanCplx},
$
or are the {\v C}ech nerve $\check{C}(\mathcal{U})$ of an open cover
$\mathcal{U} = \{U_i \to X\}_i$ of a smooth manifold $X$.
If $\mathcal{U}$ is a good cover and if $A_\bullet$ is $\mathrm{CartSp}$-acyclic
(which it is in all the examples we consider), then the function complex
$\mathbf{H}(X,A) \simeq \mathrm{Func}(\check{C}(\mathcal{U}), \mathrm{DK}(A_\bullet))$
is the traditional cocycle complex of {\v C}ech hypercohomology of $X$ with
coefficients in $A_\bullet$.

\section{Higher prequantum geometry over smooth manifolds}

We briefly review here the basic notions of higher prequantum geometry
over smooth manifolds that we will use throughout the article.
First in \ref{nPlecticManifolds} we recall the notion of
pre-$n$-plectic manifolds and their Hamiltonian vector fields
and then in \ref{PrequantumnPlecticManifolds} their pre-quantization
by {\v C}ech-Deligne cocycles.

\subsection{\texorpdfstring{$n$}{n}-Plectic manifolds and their Hamiltonian vector fields}
 \label{nPlecticManifolds}

In \cite{Baez:2008bu} the following terminology has been introduced.

\begin{definition}
\label{n-plectic_def}
  A \emph{pre-$n$-plectic manifold} $(X,\omega)$ is a smooth manifold $X$
  equipped with a closed $(n+1)$-form $\omega \in \Omega^{n+1}_{\mathrm{cl}}(X)$.
   If the contraction map $\hat{\omega} \maps TX \to \Lambda^{n} T^{\ast}X$
  is injective, then $\omega$ is called \emph{non-degenerate} or \emph{$n$-plectic}
  and $(X,\omega)$ is called an \emph{$n$-plectic manifold}.
\end{definition}
\begin{example}
  For $n = 1$ an $n$-plectic manifold is equivalently an ordinary symplectic manifold.
A compact connected simple Lie group
equipped with its canonical left invariant differential 3-form
$\omega := \langle -,[-,-]\rangle$ is a 2-plectic manifold.
  \label{CompactSimpleLieGroup2Plectic}
\end{example}

\begin{definition}
Let $(X,\omega)$ be a pre-$n$-plectic manifold. If a vector field $v$ and an $(n-1)$-form $H$ are related by
$
\iota_v\omega+dH=0
$
then we say that $v$ is a Hamiltonian field for $H$ and that $H$ is a Hamiltonian form for $v$.
  We denote by
  $
   {\mathrm{Ham}}^{n-1}(X)\subseteq \mathfrak{X}(X)\oplus\Omega^{n-1}(X)
  $
  the subspace of pairs $(v,H)$ such that $\iota_v\omega+dH=0$. We call this the
  space of \emph{Hamiltonian pairs}.
\label{VectorFieldsWithHamiltonian}
\label{HamiltonianVectorFields}
The image $\mathfrak{X}_{\mathrm{Ham}}(X)\subseteq \mathfrak{X}(X)$ of the projection ${\mathrm{Ham}}^{n-1}(X)\to \mathfrak{X}(X)$ is called the space of \emph{Hamiltonian vector fields} of $(X,\omega)$.
\end{definition}

\begin{remark}\label{remark.hamiltonian}
Given a pre-$n$-plectic manifold $(X,\omega)$
We have a short exact sequence
of vector spaces
$
0\to \Omega^{n-1}_{\mathrm{cl}}(X)\to {\mathrm{Ham}}^{n-1}(X)\to \mathfrak{X}_{\mathrm{Ham}}(X)\to 0,
$
i.e., closed $(n-1)$-forms are Hamiltonian, with zero Hamiltonian vector field.
  It is immediate from the definition that Hamilton vector fields
  preserve the pre-$n$-plectic form $\omega$, i.e.,
  $\mathcal{L}_v\omega=0$. Indeed, since $\omega$ is closed, we have
  $\mathcal{L}_v\omega=d\iota_v\omega=-d^2H_v=0$. Therefore the
  integration of a Hamiltonian vector field gives a diffeomorphism of
  $X$ preserving the pre-$n$-plectic form: a
  \emph{Hamiltonian $n$-plectomorphism}.
\end{remark}
\begin{lemma}\label{lemma.hamiltonian}
The subspace $\mathfrak{X}_{\mathrm{Ham}}(X)$ is a Lie subalgebra of $\mathfrak{X}(X)$.
\end{lemma}
\begin{proof}
Let $v_1$ and $v_2$ be Hamiltonian vector fields, and let $H_1$, $H_2$ be
their respective Hamiltonian forms. By  $\mathcal{L}_{v_{1}}\omega=0$ and
by the Cartan formulas, we get
$
\iota_{[v_{1},v_{2}]}\omega=[\mathcal{L}_{v_{1}},\iota_{v_{2}}]\omega=-\mathcal{L}_{v_{1}}dH_{2}=-d\mathcal{L}_{v_{1}}H_2=d\iota_{v_{1}}\iota_{v_{2}}\omega,
$
i.e., the vector field $[v_{1},v_{2}]$ is Hamiltonian, with Hamiltonian $\iota_{v_{1}\wedge  v_{2} }\omega$.
\end{proof}

\begin{remark}Hamiltonian vector fields on a pre-$n$-plectic manifold $(X,\omega)$ are by definition those vector fields $v$ such that $\iota_v\omega$ is exact. One may relax this condition and consider \emph{symplectic vector fields} instead, i.e., those vector fields $v$ such that $\mathcal{L}_v\omega=0$, or, equivalently, such that $\iota_v\omega$ is closed. Then the arguments in Remark \ref{remark.hamiltonian} and in Lemma \ref{lemma.hamiltonian} show that symplectic vector fields form a Lie subalgebra $\mathfrak{X}_{\mathrm{symp}}(X)$ of $\mathfrak{X}(X)$ and that $\mathfrak{X}_{\mathrm{Ham}}(X)\subseteq \mathfrak{X}_{\mathrm{symp}}(X)$ is a Lie ideal.
\end{remark}

\subsection{Prequantization of (pre-)\texorpdfstring{$n$}{n}-plectic manifolds}
 \label{PrequantumnPlecticManifolds}

The traditional notion of prequantization of a presymplectic manifold $(X,\omega)$
is equivalently a lift of the presymplectic form, regarded as a de Rham 2-cocycle,
to a degree 2 cocycle in \emph{ordinary differential cohomology}
(see, for instance \cite[Section 2.2]{Brylinski:1993}). Equivalently, this is a lift of $\omega$
to a connection $\nabla$ on a $U(1)$-principal bundle on $X$ with curvature $F_\nabla = \omega$.
Accordingly, the prequantization of a pre-$n$-plectic manifold is naturally
defined to be a lift of $\omega$ regarded as a degree $(n+1)$ cocycle in de Rham
cohomology to a cocycle of degree $(n+1)$ in
ordinary differential cohomology.
\medskip

\begin{definition}
  For $X$ a smooth manifold and $\mathcal{U} = \{U_i \to X\}$
  an open cover, we write $(\Tot(\cU,\Omega),\dt)$
  for the corresponding {\v C}ech-de Rham total complex,
  i.e.,  the cochain complex with underlying graded vector space
  $
   \Tot^{n}(\cU,\Omega)  = \bigoplus_{i+j=n} \!\check{C}^{i}(\cU,\Omega^{j})
  $
  and whose differential is given on elements
  $\bar{\theta} = \sum_{i=0}^{n} \theta^{n-i}
  $ with  $\theta^{n-i} \in \check{C}^{i}(\cU,\Omega^{n-i})$ by
$
\dt \theta^{n-i} = \delta \theta^{n-i} + (-1)^{i} d \theta^{n-i}.
$
\end{definition}

\begin{definition}
The cochain complex of sheaves
\[
C^\infty(-;U(1))\xrightarrow{d\text{log}}\Omega^1(-)\xrightarrow{d}\Omega^2(-)\xrightarrow{d}\cdots\cdots\xrightarrow{d}\Omega^{n}(-)\xrightarrow{d}\Omega^{n+1}(-)\to\cdots,
\]
with $C^\infty(-;U(1))$ in degree zero,
will be called the \emph{Deligne complex} and will be denoted by the symbol $\underline{U}(1)_{\mathrm{Del}}$. Its truncation
\[
C^\infty(-;U(1))\xrightarrow{d\text{log}}\Omega^1(-)\xrightarrow{d}\Omega^2(-)\xrightarrow{d}\cdots\cdots\xrightarrow{d}\Omega^{n}(-)\to 0\to 0\to \cdots
\]
will be denoted by $\underline{U}(1)_{\mathrm{Del}}^{\leq n}$.
\end{definition}
It follows from the above definition that a degree $n$ {\v C}ech-Deligne cocycle
in $\underline{U}(1)_{\mathrm{Del}}^{\leq n}$
is
$
\bar{A} = \sum^{n}_{i=0} A^{n-i}$, with $A^{n-i} \in
\check{C}^{i}(\mathcal{U},\Omega^{n-i})$ and $A^{0} \in \check{C}^{n}(\mathcal{U},\underline{U}(1))$,
satisfying
\begin{equation}
\begin{split}
\delta A^{n-i} &= (-1)^{i} dA^{n-i -1}, \quad i=0,\hdots,n-2\\
\delta A^{1}&= (-1)^{n-1} d\mathrm{log} A^0; \qquad
\delta A^{0} =1
\end{split}
\end{equation}
 \label{DeligneCocycle}

\begin{definition}
The \emph{$n$-stack of principal $U(1)$-$n$-bundles (or
$(n-1)$-bundle gerbes) with connection}  is the $n$-stack
presented via applying the Dold-Kan construction to the presheaf $\underline{U}(1)_{\mathrm{Del}}^{\leq n}[n]$, regarded as a presheaf of chain complexes concentrated in non-negative degree. It will be denoted by the symbol $\mathbf{B}^{n}U(1)_{\mathrm{conn}}$.
\end{definition}

The commutative diagram
\[
\scalebox{0.92}{
\xymatrix{
C^\infty(-;U(1))\ar[r]^{d\text{log}}\ar[d]&\Omega^1(-)\ar[d]\ar[r]^{d}&\cdots\ar[r]^{d}&\Omega^{n-1}(-)\ar[d]\ar[r]^{d}&\Omega^{n}(-)\ar[d]^d\\
0\ar[r]&0\ar[r]&\cdots\ar[r]&0\ar[r]&\Omega^{n+1}(-)_{\mathrm{cl}}
}
}
\]
presents the morphism of stacks
$
F:\mathbf{B}^{n}U(1)_{\mathrm{conn}}\to \Omega^{n+1}(-)_{\mathrm{cl}}
$
that maps a principal $U(1)$-$n$-bundle with connection to its \emph{curvature} $(n+1)$-form.

\begin{definition}
 Let $(X,\omega)$ be a pre-$n$-plectic manifold. A
\emph{prequantization} of $(X,\allowbreak \omega)$ is a lift
\[\scalebox{0.92}{
\xymatrix{
&\mathbf{B}^{n}U(1)_{\mathrm{conn}}\ar[d]^F\\
X\ar[r]^-{\omega}\ar[ru]^{\nabla}&\Omega^{n+1}(-)_{\mathrm{cl}}.
}
}
\]
We call the triple $(X,\omega,\nabla)$ a prequantized pre-$n$-plectic manifold.
  \label{PrequantizationOfnPlectic}
\end{definition}
Local data for a prequantization $(X,\omega,\nabla)$ are conveniently expressed in terms of the \v{C}ech-Deligne total complex.
Namely, let $\mathcal{U}$ be a good cover of $X$; then a pre-$n$-plectic structure on $X$ is the datum of a closed element $\omega$ in $\check{C}^{0}(\mathcal{U},\underline{U}(1)_{\mathrm{Del}}^{\leq n+1})$. Moreover, if $(X,\omega)$ admits a prequantization, then the datum of a prequantization is an element $A$ in $\mathrm{Tot}^n(\mathcal{U},\underline{U}(1)_{\mathrm{Del}})$ such that $d_{\mathrm{Tot}}A=\omega$.

\begin{remark}\label{remark.deligne}
It is a well know fact that $(X,\omega)$ admits a prequantization if
and only if it is an \emph{integral} pre-symplectic manifold, i.e., if and only if the closed form $\omega$ represents an integral class in
de Rham cohomology; see, e.g., \cite{Brylinski:1993}.
Indeed, since the shifted Deligne complex $\underline{U}(1)_{\mathrm{Del}}[n]$
is an acyclic resolution of the cochain complex of sheaves $\flat \mathbf{B}^nU(1)$ consisting of locally constant $U(1)$-valued functions placed in degree $-n$, we see that a pre-$n$-plectic structure $\omega$ is prequantizable if and only if $\omega$ defines the trivial class in the degree $n+1$ \v{C}ech cohomology of $X$ with coefficients in the discrete abelian group $U(1)$. By the short exact sequence of groups
$
0\to\mathbb{Z}\to \mathbb{R}\to U(1)\to 1
$
and by the \v{C}ech-de Rham isomorphism $H^n_{\mathrm{dR}}(X,\mathbb{R})\cong \check{H}^n(X,\mathbb{R})$, this is equivalent to requiring that the de Rham class of $\omega$ is an integral class.
\end{remark}

\section{The \texorpdfstring{$L_\infty$}{Loo}-algebra of local observables and its KKS \texorpdfstring{$L_\infty$}{Loo}-cocycle}
 \label{LooXomega}
To any pre-$n$-plectic manifold $(X,\omega)$ one can associate
an $L_\infty$-algebra $L_\infty(X,\omega)$, as defined in \cite{FRZ:2013,Rogers:2010nw},
which we may think of as the higher  local observables on $(X,\omega)$.
This is an $L_\infty$-extension of the Lie algebra of Hamiltonian vector fields on $(X,\omega)$
by the $(n-1)$-shifted truncated de Rham complex of $X$.
We briefly recall this construction in \ref{ThePoissonBracketAlgebra} below.

For $(V, \omega)$ an ordinary symplectic vector space, we may regard it as
a symplectic manifold that is canonically equipped with a $V$-action
by Hamiltonian vector fields, with $V$ regarded as the abelian Lie algebra
of constant (left invariant) vector fields on itself.
The evaluation map at zero $\iota_{-\wedge -}\omega\vert_0 : V \times V \to \mathbb{R}$
of the symplectic form is then a Lie algebra 2-cocycle on $V$ and hence
defines an extension of Lie algebras.
This is famous as the \emph{Heisenberg Lie algebra} extension and $\iota_{-\wedge -}\omega\vert_0$
is the \emph{Kirillov-Kostant-Souriau cocycle} that classifies it
(see example \ref{CocycleOvernPlecticVectorSpaces} below).
More generally, for any symplectic manifold, the KKS 2-cocycle classifies the underlying Lie algebra of the Poisson algebra as a central extension of the Hamiltonian vector fields  \cite{Kostant:1970,Souriau:1967}. For symplectic vector spaces, the restriction of the KKS 2-cocycle to the constant Hamiltonian vector fields is precisely the above cocycle.
We describe in
\ref{TheHeisenbergLInfinityCocycle} below a further generalization of
this to a class of $L_\infty$-algebra $(n+1)$-cocycles on Hamiltonian
vector fields over pre-$n$-plectic manifolds.
We call these the \emph{higher Kirillov-Kostant-Souriau $L_\infty$-cocycles}.
In \ref{TheKSExtension} we prove that the $L_\infty$-algebra extension
that is classified by the KKS $(n+1)$-cocycle is indeed again
the Poisson-bracket $L_\infty$-algebra of local observables.

\subsection{The \texorpdfstring{$L_\infty$}{Loo}-algebra of local observables}
 \label{ThePoissonBracketAlgebra}
We recall the construction of the $L_\infty$-algebra of local
observables associated to a pre-$n$-plectic manifold.
It is best seen in the light of the following immediate consequence of
Cartan's ``magic formula'' $\mathcal{L}_v=d \iota_v+ \iota_vd$.
\begin{lemma}\label{tech_lemma}
Let $X$ be a smooth manifold
and let  $\beta$ be an  $n$-form (not necessarily closed) on $X$.
Given $k$ vector fields $v_1,\dots,v_k$ ($k\geq 1$) on $X$, the following identity holds:
 \begin{align*}
(-1)^{k}d \iota_{v_{1} \wedge\cdots \wedge v_{k}} \beta =&
 \sum_{1 \leq i < j \leq
  k} (-1)^{i+j} \iota_{[v_{i},v_{j}] \wedge v_{1} \wedge \cdots
  \wedge \hat{v}_{i} \wedge \cdots \wedge \hat{v}_{j} \wedge \cdots \wedge v_{k}}\beta\\
&+\sum_{i=1}^{k} (-1)^{i}\iota_{ v_{1} \wedge \cdots
  \wedge \hat{v}_{i} \wedge \cdots \wedge {v}_{k}}\mathcal{L}_{v_i}\beta+ \iota_{v_{1} \wedge \cdots
  \wedge   {v}_{k}} d\beta.
\end{align*}
\end{lemma}
A  special case of the above appeared as Lemma 3.7 in \cite{Rogers:2010nw}.
We thank M.\ Zambon for pointing out to us this generalization.

\begin{proposition}[Thm.\ 5.2 \cite{Rogers:2010nw}; Thm.\ 4.7 \cite{FRZ:2013}]
\label{ham-infty}
Let $(X,\omega)$ be a pre-$n$-plectic manifold.
There exists a Lie $n$-algebra $L_{\infty}(X,\omega)$ whose underlying
chain complex is
\[
\Omega^0(X)\xrightarrow{d}\Omega^1(X)\xrightarrow{d}\cdots \xrightarrow{d}\Omega^{n-2}(X)\xrightarrow{(0,d)}{\mathrm{Ham}}^{n-1}(X)
\,,
\]
with ${\mathrm{Ham}}^{n-1}(X)$ in degree zero, and
whose multilinear brackets $l_i$ are
 \begin{equation*}
\scalebox{0.92}{$\displaystyle{ l_{1}(x)=
  \begin{cases}
  0 \oplus dx & \text{if $\deg{x}=1$},\\
  d x & \text{if $\deg{x}  > 1$,}
  \end{cases}
   \qquad\qquad
  l_{2}(x_{1},x_{2}) =
  \begin{cases}
    [v_{1},v_{2}] + \iota_{v_{1} \wedge v_{2}} \omega  &    \text{if $\deg{x_{1}}=\deg{ x_{2}}$}  = 0,\\
      0  & \text{otherwise},
  \end{cases}}$}
 \end{equation*}
 and, for $k > 2$:
 \[
 l_{k}(x_{1},\ldots,x_{k}) =
 \begin{cases}
 \vs{k}\iota_{v_{1} \wedge \cdots \wedge v_{k}} \omega & \text{if $\deg{ x_{1}}= \cdots =\deg{x_{k}} =0$}, \\
 0 & \text{otherwise},
 \end{cases}
 \]
where $x=v+\eta^\bullet$ denotes a generic element $(\eta^0,\eta^1,\dots,v+\eta^{n-1})$ in the chain complex. 
\end{proposition}
\begin{definition}\label{def.local-observables}
We call the Lie $n$-algebra $L_{\infty}(X,\omega)$ defined in the
statement of Proposition \ref{ham-infty}  the \emph{$L_\infty$-algebra of local observables}
on $(X,\omega)$.
\label{TheLooXOmega}
\end{definition}

\begin{remark}
  The projection map of def. \ref{HamiltonianVectorFields}
  uniquely extends to a morphism of $L_\infty$-algebras of the form
  $L_\infty(X,\omega)
	  \xrightarrow{\pi_L}
	  \mathfrak{X}_{\mathrm{Ham}}(X)
	\,,
  $
i.e., local observables of $(X,\omega)$ cover Hamiltonian vector fields. Below in \ref{TheHeisenbergLInfinityCocycle}
we turn to the classification of this map by an $L_\infty$-algebra cocycle.

  \label{MapFromLocalObservablesToHamiltonianVectorFields}
\end{remark}

\begin{example}
 If $n = 1$ then $(X,\omega)$ is a pre-symplectic manifold, the
 chain complex underlying $L_\infty(X,\omega)$ is
$
{\mathrm{Ham}}^0(X)=\{v+H\in \mathfrak{X}(X)\oplus C^\infty(X;\mathbb{R})\,|\, \iota_v\omega+dH=0\},
$
\allowbreak and the Lie bracket is
$
[v_1+H_1,v_2+H_2]= [v_1,v_2]+\iota_{v_1\wedge v_2}\omega.
$
If moreover $\omega$ is non-degenerate so that $(X,\omega)$ is symplectic,
then the projection $v+H \mapsto H$ is a linear isomorphism
${\mathrm{Ham}}^0(X) \stackrel{\simeq}{\to} C^\infty(X;\mathbb{R})$.
It is easy to see that under this isomorphism
$L_\infty(X,\omega)$ is the underlying Lie algebra of the usual
Poisson algebra of functions.  See also Prop.\ 2.3.9 in \cite {Brylinski:1993}.
  \label{OrdinaryPoissonBracket}
\end{example}

\subsection{The Kirillov-Kostant-Souriau \texorpdfstring{$L_\infty$}{Loo}-algebra cocycle}
 \label{TheHeisenbergLInfinityCocycle}

 Here we present an $L_\infty$-algebra cocycle 
 on the Lie algebra of Hamiltonian vector fields on
 a pre-$n$-plectic manifold, which generalizes the traditional
KKS cocycle and the Heisenberg cocycle to higher
 geometry.

\begin{definition}
For $X$ a smooth manifold, denote by $\mathbf{B}\mathbf{H}(X,\flat\mathbf{B}^{n-1}\mathbb{R})$
the abelian Lie $(n+1)$-algebra given by the chain complex
$
\Omega^0(X)\!\xrightarrow{d}\Omega^1(X)\!\xrightarrow{d}\cdots \!\xrightarrow{d}\Omega^{n-1}(X)\!\xrightarrow{d}d\Omega^{n-1}(X),
$
with $d\Omega^{n-1}(X)$ in degree zero.
\label{ResolvedDeloopedShiftedTruncatedDeRhamComplex}
\end{definition}
\begin{remark}
  The complex of def. \ref{ResolvedDeloopedShiftedTruncatedDeRhamComplex}
  serves as a resolution of the cocycle complex
$
\Omega^0(X)\xrightarrow{d}\Omega^1(X)\xrightarrow{d}\cdots \xrightarrow{d}\Omega_{\mathrm{cl}}^{n-1}(X)\xrightarrow{} 0
\,,
$
{for the de Rham cohomology of $X$ up to degree $n-1$}
once delooped (i.e., shifted).
  \label{CoefficientsOvernConnectedSpace}
\end{remark}
\begin{proposition}
  \label{prop.just-above}
  Let $(X,\omega)$ be a pre-$n$-plectic manifold. The multilinear maps
\begin{align*}
  &\omega_{[1]} : v\mapsto -\iota_v\omega;
  \qquad \qquad
  \omega_{[2]} : v_1\wedge v_2\mapsto \iota_{v_1\wedge v_2}\omega;
  \qquad \qquad
     \cdots\\
     &
  \omega_{[n+1]} : v_1\wedge v_2\wedge\cdots  v_{n+1}\mapsto -(-1)^{\binom{n+1}{2}}
\iota_{v_1\wedge v_2\wedge \cdots \wedge v_{n+1}}\omega
\end{align*}
define an $L_\infty$-morphism
$
  \omega_{[\bullet]}
  :
\mathfrak{X}_{\mathrm{Ham}}(X)\xrightarrow{} \mathbf{B} \mathbf{H}(X,\flat\mathbf{B}^{n-1}\mathbb{R})\,,
$
and hence an $L_\infty$-alge\-bra $(n+1)$-cocycle on the Lie algebra of
Hamiltonian vector fields, def. \ref{HamiltonianVectorFields}, with
values in the abelian $(n+1)$-algebra of
def. \ref{ResolvedDeloopedShiftedTruncatedDeRhamComplex}.
\end{proposition}
\begin{proof}
  First notice that the underlying map on chain complexes is indeed
  well defined: by definition of Hamiltonian vector fields, if $v$ is
  Hamiltonian, then there exists an $(n-1)$-form $H$ such that
  $\iota_v\omega+dH=0$ and so $\omega_{[\bullet]}$ takes
  values in $\mathbf{B}
  \mathbf{H}(X,\flat\mathbf{B}^{n-1}\mathbb{R})$.
 In general, an $L_\infty$-algebra morphism $f \maps \mathfrak{g} \to \mathfrak{h}$
 from a Lie algebra $\mathfrak{g}$ to an abelian Lie $(n+1)$-algebra $\mathfrak{h}$ is equivalently a
  collection of linear maps
  $\{ f_k : \wedge^k \mathfrak{g}_\bullet \to \mathfrak{h}_{\bullet} \}_{k = 1}^{n+1}$
  with $\deg{f_{k}} =k-1$ and such that the following holds for all $k \geq 1$
  \begin{equation*}    \label{MorphismToAbelianLooAlgebra}
   d_{\mathfrak h}f_k(v_1^{}\wedge\cdots\wedge v_k^{})=
   \sum_{i<j}(-1)^{i+j+1} f_{k-1}([v_i,v_j]^{}_{\mathfrak g}\wedge
    v_1\wedge\cdots
    \wedge\widehat{v_i}\wedge\cdots\wedge\widehat{v_j}\wedge
    \cdots\wedge v_{k}^{}).
   \end{equation*}
Therefore, checking that $\omega_{[\bullet]}$
is an $L_\infty$-morphism reduces to checking the identities
\[
d\iota_{v_1^{}\wedge\cdots\wedge v_k^{}}\omega=
\sum_{i<j} (-1)^{i+j+k} \iota_{[v_i,v_j]\wedge
v_1\wedge\cdots
\wedge\widehat{v_i}\wedge\cdots\wedge\widehat{v_j}\wedge
\cdots\wedge v_{k+1}^{}}\omega.
\]
These are satisfied -- since the $\omega$ is closed and the $v_i$ are Hamiltonian -- by Lemma \ref{tech_lemma}.
\end{proof}

\begin{definition}
The degree $(n+1)$  \emph{higher Kirillov-Kostant-Souriau $L_\infty$-cocycle}
as\-sociated to the pre-$n$-plectic manifold $(X,\omega)$ is the $L_{\infty}$-morphism
  $
	  \mathfrak{X}_{\mathrm{Ham}}(X)
	  \xrightarrow{\omega_{[\bullet]}}
	  \mathbf{B}\mathbf{H}(X,\allowbreak \flat\mathbf{B}^{n-1}\mathbb{R})
	  $
  given in Prop.\ \ref{prop.just-above}.
  \label{HigherHeisenbergCocycle}
\end{definition}
If $\rho \colon \mathfrak{g} \to \mathfrak{X}_{\mathrm{Ham}}(X)$ is
an $L_\infty$-morphism encoding an action of an $L_\infty$-algebra
$\mathfrak{g}$ on $(X,\omega)$ by Hamiltonian vector fields, then we call
the composite $\rho^* \omega_{[\bullet]}$ the corresponding
\emph{Heisenberg $L_\infty$-algebra cocycle}. This terminology is motivated by the
following example \ref{CocycleOvernPlecticVectorSpaces}. Further discussion
of this aspect is below in Section \ref{TheHeisenbergExtension}.

\begin{example}
   Let $V$ be a vector space equipped with a skew-symmetric multilinear
  form $\omega: \wedge^{n+1} V \to \R$. Since $V$ is an abelian Lie
  group, we obtain via left-translation of $\omega$ a unique closed invariant form,
  which we also denote as $\omega$. By identifying $V$ with left-invariant
  vector fields on $V$, the Poincare lemma implies that we have a
  canonical inclusion
$
    j_V : V \hookrightarrow \mathfrak{X}_{\mathrm{Ham}}(V)
  $
  of $V$ regarded as an abelian Lie algebra into the Hamiltonian vector fields on $(V,\omega)$
  regarded as a pre $n$-plectic manifold.
  Since $V$ is contractible as a topological manifold, we have,
  by remark \ref{CoefficientsOvernConnectedSpace}, a quasi-isomorphism
  $	  \mathbf{B}\mathbf{H}(V;\flat\mathbf{B}^{n-1}\mathbb{R})
	  \xrightarrow{\simeq}
	  \mathbb{R}[n]
  $
  of abelian $L_\infty$-algebras, given by evaluation at $0$.
  Under this equivalence the restriction of the
  $L_\infty$-algebra cocycle $\omega_{[\bullet]}$ of def. \ref{HigherHeisenbergCocycle}
  along $j_V$ is an $L_\infty$-algebra map of the form
  $
    j_V^* \omega_{[\bullet]}
	:
	  V
	  \to
	  \mathbb{R}[n]
	$
  whose single component is the linear map
  $
    \iota_{(-)} \omega : \wedge^{n+1} V \to \mathbb{R}
	\,.
  $
  For $n=1$ and $(V,\omega)$ an ordinary symplectic vector space
the map
  $\iota_{(-)}\omega : V \wedge V \to \mathbb{R}$ is the traditional
  \emph{Heisenberg cocycle}.
  \label{CocycleOvernPlecticVectorSpaces}
\end{example}

\begin{remark}
The KKS $(n+1)$-cocycle has a natural geometric origin as
the Lie differentiation of a morphism of higher smooth groups
canonically arising in higher geometric prequantization, see
\cite{companionArticle}. This can be seen as a deeper conceptual
justification for def.\ \ref{HigherHeisenbergCocycle}.
\end{remark}

\subsection{The Kirillov-Kostant-Souriau \texorpdfstring{$L_\infty$}{Loo}-extension}
 \label{TheKSExtension}

 Using the results presented above, we can now state and prove the
 main theorem of this section.

\begin{theorem}
Given a pre-$n$-plectic manifold $(X,\omega)$, the
higher KKS $L_\infty$-cocycle $\omega_{[\bullet]}$
(def. \ref{HigherHeisenbergCocycle}) and the projection map
$\pi_L \maps L_{\infty}(X,\omega) \to \Xham(X)$
(remark \ref{MapFromLocalObservablesToHamiltonianVectorFields}) form a homotopy fiber sequence of
$L_\infty$-algebras, i.e., fit into a homotopy pullback
 diagram of the form
\[\scalebox{0.92}{
\xymatrix{
  L_\infty(X,\omega)
    \ar[d]^{\pi_L}
    \ar[r]
	&
    0\ar[d]
  \\
  \mathfrak{X}_{\mathrm{Ham}}(X)
    \ar[r]^-{\omega_{[\bullet]}}
	&
  \mathbf{B}\mathbf{H}(X,\flat\mathbf{B}^{n-1}\mathbb{R}).
}
}
\]
\label{TheExtensionClassifiedByTheCocycle}
\end{theorem}
\begin{proof}
By theorem \ref{DomenicosLemma} it is sufficient to
replace the map of chain complexes $0 \to \mathbf{B}\mathbf{H}(X,\flat\mathbf{B}^{n-1}\mathbb{R})$
by any degreewise surjection $K\xrightarrow{\pi_R} \mathbf{B}\mathbf{H}(X,\flat\mathbf{B}^{n-1}\mathbb{R})$
out of an exact chain complex $K$,
such that its pullback along $\omega_1$
is isomorphic to the underlying chain complex of $L_\infty(X,\omega)$
and then to show that the $L_\infty$-structure of $L_\infty(X,\omega)$
sits compatibly in the resulting square diagram.
We take $K$ to be the cone of the identity of the chain complex
$
\Omega^0(X)\xrightarrow{d}\Omega^1(X)\xrightarrow{d}\cdots \xrightarrow{d}\Omega^{n-1}(X)
$
with $\Omega^{n-1}(X)$ in degree zero, and take $\pi_R$ to be the chain map
given by the vertical arrows in the following diagram:
\[
\scalebox{0.885}{
\xymatrix@R=-2pt{\Omega^0(X)\ar[r]^{d}\ar[ddr]^{\mathrm{id}}\ar[dddddd]_{\mathrm{id}} & \Omega^1(X)\ar[r]^{d}\ar[ddr]^{\mathrm{id}} & \Omega^2(X) \ar[r]^{d}\ar[ddr]^{\mathrm{id}}
&{\phantom{mmmmm\text{$\Omega^1_1$}}}\ar[r]^{d}\ar[ddr]^{\mathrm{id}}&\Omega^{n-1}(X)\ar[rdd]^{\mathrm{id}}&&
\\ & \oplus & \oplus & \cdots&\oplus&\\
& \Omega^0(X)\ar[r]^{d}\ar[dddd]_{\mathrm{id}\oplus0}& \Omega^1(X) \ar[r]^{d}&{\phantom{\text{$\Omega^1$}mmmmm}}\ar[r]^{d}&\Omega^{n-2}(X)\ar[dddd]_{\mathrm{id}\oplus0}\ar[r]^-{d}&\Omega^{n-1}(X)\ar[dddd]_{d}\\
\\
{\phantom{{\displaystyle{\int}}}}
\\
{\phantom{{\displaystyle{\int}}}}
\\
& & & & &
\\
\Omega^0(X)\ar[r]^-{d}&\Omega^1(X)\ar[r]^-{d}&\Omega^2(X)\ar[r]^-{d}&{\phantom{mm}}\cdots{\phantom{mm}} \ar[r]^-{d}&\Omega^{n-1}(X)\ar[r]^-{d}&d\Omega^{n-1}(X)
}
}
\]
By inspection and comparison with prop. \ref{ham-infty} it is easy to see that the fiber product of chain complexes of $K$ and $\mathfrak{X}_{\mathrm{Ham}}(X)$ over $\mathbf{B}\mathbf{H}(X,\flat\mathbf{B}^{n-1}\mathbb{R})$
is the chain complex $L_\infty(X,\omega)_\bullet$
that underlies the $L_\infty$-algebra of local observables:
 \[
 \scalebox{0.92}{
   \raisebox{20pt}{
 \xymatrix{
    L_\infty(X,\omega)_\bullet
	\ar[r]^{f_1}\ar[d]_{\pi_L}&K\ar[d]^{\pi_R}
    \\
    \mathfrak{X}_{\mathrm{Ham}}(X)
	  \ar[r]^-{\omega_{[1]}}
	  &
	\mathbf{B}\mathbf{H}(X,\flat\mathbf{B}^{n-1}\mathbb{R}),
 }
 }
 \,,
 }
 \]
 where  $f_1$ is the morphism
\[
f_1:
v+\eta^\bullet\mapsto
\left(
\begin{matrix}
0&0&0&\cdots &0&0&\\
&\eta^0 & \eta^1 &\cdots &\eta^{n-3}&\eta^{n-2}&\eta^{n-1}
\end{matrix}
\right).
\]
As we already observed in remark \ref{MapFromLocalObservablesToHamiltonianVectorFields},
 the chain map underlying $\pi_L$ uniquely extends to an $L_\infty$-algebra
 morphism.
 Therefore to complete the proof, it is sufficient to show that we can lift the horizontal chain map $f_1$ above to a morphism of
 $L_\infty$-algebras which makes the diagram
 \[
 \scalebox{0.92}{
 \xymatrix{
L_\infty(X,\omega)\ar[r]^{f}\ar[d]_{\pi_L}&K\ar[d]^{\pi_R}\\
\mathfrak{X}_{\mathrm{Ham}}(X)\ar[r]^-{\omega_{[\bullet]}}&\mathbf{B}\mathbf{H}(X,\flat\mathbf{B}^{n-1}\mathbb{R}),
 }
 }
 \]
  commute.
This is easily realized by defining the  ``Taylor coefficients''  of $f$ for $k\geq 2$  to be the degree $(k-1)$ maps $f_k:\wedge^k L_\infty(X,\omega)\to K$ given by
 \[
  f_k:(v_1+\eta_1^\bullet)\wedge \cdots \wedge (v_k+\eta_k^\bullet)\mapsto\scalebox{0.87}{$ \left(
\begin{matrix}
0&0&0&&&\cdots &0&0&\\
&0 & 0 &\cdots &0&\vs{k}\iota_{v_1\wedge\cdots\wedge v_k}\omega&0&\cdots&0
\end{matrix}
\right).
 $}
 \]
\end{proof}

\subsection{The Heisenberg \texorpdfstring{$L_\infty$}{Loo}-extension}
 \label{TheHeisenbergExtension}

If a Lie algebra $\mathfrak{g}$ acts on an
$n$-plectic manifold by Hamiltonian vector fields,
then the KKS $L_\infty$-extension of $\mathfrak{X}_{\mathrm{Ham}}(X)$,
discussed above in \ref{TheKSExtension},
restricts to an $L_\infty$-extension of $\mathfrak{g}$. This is a generalization
{of Kostant's construction \cite{Kostant:1970} of central extensions of Lie algebras
  to the context of $L_\infty$-algebras. Perhaps the most famous of
  these central extensions is the Heisenberg Lie algebra, which is the
  inspiration behind the following terminology:}

\begin{definition}
  Let $(X,\omega)$ be a pre-$n$-plectic manifold and let $\rho :
  \mathfrak{g} \to \mathfrak{X}_{\mathrm{Ham}}(X)$ be a Lie algebra
  homomorphism encoding an action of $\mathfrak{g}$ on $X$ by
  Hamiltonian vector fields. The corresponding \emph{Heisenberg
    $L_\infty$-algebra extension} $\mathfrak{heis}_\rho(\mathfrak{g})$
  of $\mathfrak{g}$ is the extension classified by the composite
  $L_{\infty}$-morphism $\omega_{[\bullet]}\circ \rho$,  i.e., the
  homotopy pullback on the left of
  $$
  \scalebox{0.92}{
    \raisebox{20pt}{
    \xymatrix{
	  \mathfrak{heis}_\rho(\mathfrak{g})
	  \ar[r]
	  \ar[d]
	  &
	  L_\infty(X,\omega)
	  \ar[r]
	  \ar[d]
	  &
	  0
	  \ar[d]
	  \\
	  \mathfrak{g}
	  \ar[r]^-\rho
	  &
	  \mathfrak{X}_{\mathrm{Ham}}(X)
	  \ar[r]^-{\omega_{[\bullet]}}
	  &
	  \mathbf{B}\mathbf{H}(X,\flat \mathbf{B}^{n-1}\mathbb{R})
	}
	}
	\,.}
  $$
\end{definition}
\begin{remark}
  It is natural to call an $L_\infty$-morphism with values in the
  $L_\infty$-algebra of observables of a pre-$n$-plectic manifold
  $(X,\omega)$ an `$L_\infty$ {comoment map}', which
  generalizes the familiar notion in symplectic geometry. Hence, one
  could say that an action $\rho$ of a Lie algebra $\mathfrak{g}$ on a
  pre-$n$-plectic manifold $(X,\omega)$ via Hamiltonian vector fields
  naturally induces such a co-moment map from the Heisenberg
  $L_\infty$-algebra $\mathfrak{heis}_\rho(\mathfrak{g})$.
 \end{remark}

\begin{example}
  For $(V,\omega)$ a symplectic vector space regarded as a symplectic
  manifold, the translation action of $V$ on itself is {via
    Hamiltonian vector fields} (see example \ref{CocycleOvernPlecticVectorSpaces}). If one denotes by
  $
    \rho : V \to \mathfrak{X}_{\mathrm{Ham}}(X)
  $
 this action, then the induced
  Heisenberg $L_\infty$-extension is the traditional Heisenberg Lie algebra.
\end{example}
\begin{example}
  Let  $G$ be a {(connected)} compact simple Lie group, regarded as a 2-plectic manifold
  with its canonical 3-form $\omega := \langle -,[-,-]\rangle$
  as in example \ref{CompactSimpleLieGroup2Plectic}.  {The
    infinitesimal generators of the action of $G$ on
    itself by right translation are the left invariant vector fields $\gg$, which are
    Hamiltonian.}
{We have $H^{1}_{\mathrm{dR}}(G) \cong
  H^{1}_{\mathrm{CE}}(\gg,\R)=0$, and therefore a weak equivalence:}
  $
	  \mathbf{B} \mathbf{H}(G, \flat\mathbf{B} \mathbb{R})
	  \xrightarrow{\simeq}
	  \mathbb{R}[2]	
  $
  {given by the evaluation at the identity element of $G$.}
  The resulting composite cocycle
  $$
    \langle -,[-,-]\rangle
	:
    \xymatrix{
	  \mathfrak{g}
	  \ar[r]^-\rho
	  &
	  \mathfrak{X}_{\mathrm{Ham}}(X)
	  \ar[r]^-{\omega_{[\bullet]}}
	  &
	  \mathbb{R}[2]
	}
  $$
  is exactly the Lie algebra 3-cocycle that classifies the String Lie-2-algebra.
  By theorem \ref{DomenicosLemma} the String Lie 2-algebra is the homotopy fiber
  of this cocycle, in that we have a homotopy pullback square of $L_\infty$-algebras
  $$
  \scalebox{0.92}{
    \raisebox{20pt}{
    \xymatrix{
	  \mathfrak{string}_{\mathfrak{g}}
	  \ar[r]
	  \ar[d]
	  &
	  0
	  \ar[d]
	  \\
	  \mathfrak{g}
	  \ar[r]^-{\langle -,[-,-]\rangle}
	  &
	  \mathbb{R}[2]
	}}
	\,.}
  $$
  Hence, the String Lie 2-algebra $ \mathfrak{string}_{\mathfrak{g}}$  is the Heisenberg Lie 2-algebra
  of the 2-plectic manifold $(G,\langle -,[-,-]\rangle)$ with its canonical
  $\mathfrak{g}$-action $\rho$, i.e.,   $
    \mathfrak{heis}_\rho(\mathfrak{g}) \simeq \mathfrak{string}_{\mathfrak{g}}
	\,.
  $
  \label{StringAsHeisenberg}
  The relationship between $\mathfrak{string}_{\mathfrak{g}}$ and $L_{\infty}(G,\omega)$
  was first explored in \cite{RogersString}.
\end{example}

\section{The dg-Lie algebra of infinitesimal quantomorphisms}
 \label{ThedgLieAlgOfInfinitesimalQuantomorphisms}
 \label{HigherQuantomorphismGroupsAndTheirLieDerivation}

 The $L_\infty$-algebra $L_\infty(X,\omega)$
 discussed above in Section \ref{LooXomega} has the nice property that the definition
 of its brackets generalizes the definition of the traditional Poisson bracket
 in an elegant way. We now present another $L_\infty$-algebra
 that looks a little less elegant in components, but has
 a more manifest conceptual interpretation, namely as the dg-Lie algebra
 of infinitesimal automorphisms of a $U(1)$-$n$-bundle with connection
that cover the diffeomorphisms of the base.
\nbbt{A main result of this section is Thm.\
  \ref{theorem.main-theorem}, which establishes a weak equivalence
  between the aforementioned dg Lie algebra and $L_{\infty}(M,\omega)$.}

 \subsection{Quantomorphism \texorpdfstring{$n$}{n}-groups}

 Since, by definition, a prequantization of a pre-$n$-plectic manifold
 $(X,\omega)$ is a morphism of higher stacks $X\to \mathbf{B}^{n}U(1)_{\mathrm{conn}}$,
 a prequantized pre-$n$-plectic manifold is naturally an object in the
overcategory (or ``slice topos'') $\mathbf{H}_{/\mathbf{B}^{n}U(1)_{\mathrm{conn}}}$.
This leads to the following definition.
\begin{definition}
  Let $\nabla_0, \nabla_1: X\to \mathbf{B}^nU(1)_{\mathrm{conn}}$ be
  two morphisms representing (or ``modulating'') principal
  $U(1)$-$n$-bundles with connection on $X$. A \textit{1-morphism}
  $(\phi,\eta) \maps \allowbreak  \nabla_0 \to \nabla_1$ in
  $\mathbf{H}_{/\mathbf{B}^{n}U(1)_{\mathrm{conn}}}$ is a homotopy
  commutative diagram of the form
  \[
  \scalebox{0.92}{
    \raisebox{20pt}{ \xymatrix{
	   X \ar[rr]^{\phi}_{\phantom{m}}="s" \ar[dr]_{\nabla_0}^{\ }="t" && X \ar[dl]^{\nabla_1}
	   \\
	   & \mathbf{B}^{n}U(1)_{\mathrm{conn}}
	   \ar@{=>}^\eta "s"; "t"
	}}.}
 \]
 A \textit{2-morphism} $(k,h) \maps (\phi_1, \eta_1) \to (\phi_2,
   \eta_2)$ is only between 1-morphisms such that $\phi_1 = \phi_2$
 and is given by a homotopy commutative diagram of the form
	 	 $$
		 \scalebox{0.92}{
		  \raisebox{20pt}{
	   \xymatrix{
	       X \ar[rr]^{\phi_1=\phi_2}_{\ }="s"
		     \ar@/^.4pc/[dr]|{\nabla_0}^{\ }="t"_{\ }="s2"
			 \ar@/_1.3pc/[dr]_{\nabla_0}^{\ }="t2"
			 && X \ar[dl]^{\nabla_1}
		   \\
		   & \mathbf{B}^nU(1)_{\mathrm{conn}}
		   \ar@{=>}^{\eta_1} "s"; "t"
		   \ar @{=>}_{k} "s2"; "t2"
	   }}
	   \,,
	   }
	 $$
 where one has the (undisplayed) 2-arrow $\eta_2$ on the back face of
 the diagrams and an undisplayed 3-arrow $h:k\circ \eta_1\to \eta_2$ decorating the bulk of the 3-simplex. Higher morphisms are defined similarly.
 \end{definition}
 \begin{remark}
   Since we are dealing with a commutative diagram of morphisms
   between (higher) stacks, we have the homotopy
   $\eta$ appearing here as part of the data of the commutative
   diagram defining a 1-morphism. In particular, isomorphisms (or
   better, equivalences) between $\nabla_0$ and $\nabla_1$ will be
   pairs $(\phi,\eta)$ consisting of a diffeomorphism $\phi \colon X
   \to X$ and a gauge transformation of higher connections
  $
   \eta :
	   \phi^* \nabla_1 \xrightarrow{\simeq}  \nabla_0.
	     $
   In particular, for the 1-plectic (i.e., symplectic) case,
    $\nabla_0$ and $\nabla_1$ correspond to principal $U(1)$-bundles
    with connection. If $X$ is compact, then the 1-morphisms between them correspond to
    ``strict contactomorphisms'' $(P_0,A_0) \to (P_1,A_1)$ between the
    total spaces of the bundles with their connection 1-forms $A_{i} \in
    \Omega^{1}(P_{i};\R)$ regarded as ``regular'' contact forms.
    If $\nabla = \nabla_0 = \nabla_1$ and
    $\nabla$ is regarded as the prequantization of its curvature, i.e.,
    the symplectic 2-form $\omega$, then such a contactomorphism is often called a
    \emph{quantomorphism} in the geometric quantization literature.
  \end{remark}

The automorphisms $\mathbf{Aut}_{/_{\mathbf{B}^nU(1)_{\mathrm{conn}}}}(\nabla)$
  of any object $\nabla \in \mathbf{H}_{/\mathbf{B}^{n}U(1)_{\mathrm{conn}}}$ form an
  ``$n$-group" (see, for example, Sec.\ 2.3
  of \cite{NikolausSchreiberStevensonI}).
And so, motivated by the terminology used in the above remark, we introduce the following definition.
\begin{definition}
  Let $\nabla: X \to \mathbf{B}^nU(1)_{\mathrm{conn}}$ be a morphism
  modulating a $U(1)$-$n$-bun\-dle with connection. The \emph{quantomorphism
    $n$-group} of $\nabla$, denoted $\mathbf{QuantMorph}(\nabla)$,
  is the automorphism $n$-group
  $\mathbf{Aut}_{/_{\mathbf{B}^nU(1)_{\mathrm{conn}}}}(\nabla)$
  equipped with its natural smooth structure.
  \end{definition}
\begin{remark}\label{remark.concretification}
  In the above definition we described
  $\mathbf{QuantMorph}(\nabla)$ as a ``smooth $n$-group''. In order
  to make this precise, we need to say what a smooth family of
  automorphisms is. This is systematically done by working with smooth
  families from the very beginning, i.e., by replacing the hom-spaces
  $\mathbf{H}(X,\mathbf{B}^nU(1)_{\mathrm{conn}})$ by what we call the
  ``concretification'' of the internal homs (the higher mapping
  stacks) $[X,\mathbf{B}^nU(1)_{\mathrm{conn}}]$.  See Sec.\ 2.3.2 of
  \cite{companionArticle} for precise discussion of
  this aspect. The intuition behind this smooth structure - which
    is all that we need for our purposes here - is that
    all local bundle data depend smoothly on a parameter varying
    in the base.
\end{remark}

\subsection{Infinitesimal quantomorphisms as a strict model for the
  \texorpdfstring{$L_\infty$}{Loo}-algebra of observables} \label{inf_quant_sec}
{Since the quantomorphism $n$-group $\mathbf{QuantMorph}(\nabla)$ is
equipped with a smooth structure, it has a notion of ``tangent
vectors''. Roughly speaking, these correspond to maps out of the
formal infinitesimal interval, $\mathrm{Spec}
\bigl(\R[\epsilon]/(\epsilon)^2 \bigr) \to
\mathbf{QuantMorph}(\nabla)$. So it is not surprising that there is
also an abstract notion of ``Lie differentiation'' in this context which,
when applied to the smooth $n$-group $\mathbf{QuantMorph}(\nabla)$,
produces not a Lie algebra, but rather a Lie $n$-algebra, which will be denoted 
$\mathrm{Lie}(\mathbf{Quant}$-$\mathbf{Morph}(\nabla))$.
 (See Sec.\ 3.10.9
and Sec.\ 4.5.1.2 in \cite{dcct} for more details on Lie differentiation). }

The defining equations of $\mathrm{Lie}(\mathbf{QuantMorph}(\nabla))$
can be conceptually described as the infinitesimal
 versions of the defining equations for the quantomorphism
 $n$-group. In particular, a degree zero element in
 $\mathrm{Lie}(\mathbf{QuantMorph}(\nabla))$ will be an infinitesimal
 version of a pair $(\phi : X \xrightarrow{\sim} X,\;h: \phi^* \nabla
 \xrightarrow{\sim} \nabla)$, i.e., a pair $(v,b)$ consisting of
a vector field $v$ on $X$ and an ``infinitesimal homotopy'' $b$ such that
$
b \colon \mathcal{L}_v\nabla \to 0,
$
where $\mathcal{L}_v$ is the Lie derivative along $v$. Degree 1
elements in $\mathrm{Lie}(\mathbf{QuantMorph}(\nabla))$ will be
homotopies between the $b$'s, and so on.
{The notion of taking the Lie derivative of a morphism of higher stacks
may give pause, but it has an obvious interpretation if we represent
the map $\nabla:X\to \mathbf{B}^nU(1)_{\mathrm{conn}}$ as a {\v
  C}ech-Deligne cocycle $\bar{A}$ on $X$
(def. \ref{DeligneCocycle}). In this context, $\mathcal{L}_v\nabla$
corresponds to the usual Lie derivative $\mathcal{L}_v \bar{A}$ for
vector fields acting on local differential forms. Moreover, in this
context the Dold-Kan correspondence tells us that, for example, an
infinitesimal homotopy $b \colon \mathcal{L}_v\nabla \to 0$ is simply
an element $\bar{\theta}$ of the total \v{C}ech-de Rham complex
$\bigl(\Tot^{\bullet}(\cU,\Omega), \dt \bigr)$ satisfying
$
\mathcal{L}_v \bar{A} = \dt \bar{\theta}.
$
}
The above discussion is the intuition behind the
following:
 \begin{defprop}\label{def.dg-lie-quant}
{Let $X$ be a smooth manifold and $n \in \mathbb{N}$.
If $\bar{A}$ is a {\v C}ech-Deligne
$n$-cocycle on $X$ relative to some cover $\mathcal{U}$, then
the \textit{dg Lie algebra of infinitesimal quantomorphisms}
$
\mathrm{dgLie}_{\mathrm{Qu}}(X,\bar{A})
$
is the strict Lie $n$-algebra whose underlying complex is
\begin{equation*}
\begin{split}
\mathrm{dgLie}_{\mathrm{Qu}} (X,\bar{A})^{0} &= \{ v +  \bar{\theta} \in \X(M) \oplus
\Tot^{n-1}(\cU,\Omega) ~\vert ~  \L_{v} \bar{A} = \dt \bar{\theta} \},\\
\mathrm{dgLie}_{\mathrm{Qu}} (X,\bar{A})^{i} &= \Tot^{n-1 - i}(\cU,\Omega) \quad\text{for}\quad  1\leq i\leq n-1,
\end{split}
\end{equation*}
with differential
\[
\scalebox{0.86}{$
\mathrm{dgLie}_{\mathrm{Qu}} (X,\bar{A})^{n-1} \xto{\dt}
\mathrm{dgLie}_{\mathrm{Qu}} (X,\bar{A})^{n-2} \xto{\dt} \cdots \xto{\dt}
\mathrm{dgLie}_{\mathrm{Qu}} (X,\bar{A})^{1} \xto{0\oplus \dt} \mathrm{dgLie}_{\mathrm{Qu}} (X,\bar{A})^{0},
$}
\]
and whose graded Lie bracket is the semidirect product bracket for the
Lie algebra of vector fields acting on differential forms by Lie derivative:
\begin{equation}
\begin{split}
&\bbrac{v_1 + \bar{\theta}_{1}}{v_2 + \bar{\theta}_{2}}  = [v_1,v_2] +
\L_{v_{1}} \bar{\theta}_{2} -\L_{v_{2}} \bar{\theta}_{1};\\
&\bbrac{v + \bar{\theta}}{\bar{\eta}}  = -\bbrac{\bar{\eta}}{v + \bar{\theta}}= \L_{v} \bar{\eta};\qquad\qquad
\bbrac{\bar{\eta}}{\bar{\eta}}=0.
\end{split}
\end{equation}
}
 \label{dgLieXOmega}
\end{defprop}

{
The next theorem reveals the relationship between the above dgla of
infinitesimal quantomorphisms and the $L_{\infty}$-algebra of local
observables. It is the higher analogue of the well-known fact in
traditional prequantization that the underlying Lie algebra of the
Poisson algebra on a prequantized symplectic manifold is isomorphic to
the Lie algebra of $U(1)$-invariant connection-preserving vector fields on the total
space of the prequantum bundle.}

\begin{theorem}
 \label{theorem.main-theorem}
 Let $(X,\omega)$ be an integral pre-$n$-plectic manifold (def.\ \ref{n-plectic_def}), $\mathcal{U}
 $ a good open cover of $X$,
 and $\nabla$  a prequantization of $(X,\omega)$ (def. \ref{PrequantizationOfnPlectic})
 presented by a {\v C}ech-Deligne cocycle $\bar{A} = \sum_{i=0}^{n}A^{n-i}$ in $\mathrm{Tot}^n(\mathcal{U},\underline{U}(1)_{\mathrm{Del}}^{\leq n})$. There
exists an $L_\infty$-quasi-iso\-morphism
$
  f:
    L_\infty(X,\omega)
	\xrightarrow{\simeq}
	\mathrm{dgLie}_{\mathrm{Qu}} (X,\bar{A})
  \,
$
between the $L_\infty$-algebra of local observables
(def.\ \ref{TheLooXOmega}) and the dgla of infinitesimal quantomorphisms
(def.\ \ref{dgLieXOmega}), whose linear term is
\begin{equation*}
f_{1}(x) =
\begin{cases}
v - H\vert_{U_{\alpha}} +\sum_{i=0}^{n} (-1)^{i} \iota_vA^{n-i} & \qquad\forall x =  v +  H \in  {\mathrm{Ham}}^{n-1}(X)\\
-  x\vert_{U_{\alpha}}& \qquad \forall x \in \Omega^{n-1-i}(X) \quad i \geq 1\end{cases}
\end{equation*}
and whose higher components $f_k$ are explicitly determined by Eq.\ \ref{struct_maps}.
\end{theorem}

\begin{proof}
  The linear morphism $f_1$ is essentially the familiar quasi-isomorphism between the de Rham complex and
  the total \v{C}ech-de Rham complex. Proving that $f_1$ lifts to an $L_{\infty}$-morphism and explicitly determining the higher components of this $L_\infty$-morphism is a lengthy
  but straightforward computation. We report it in Appendix
  \ref{appendix.proof}.
\end{proof}

\begin{remark}
By homological perturbation theory \cite{Huebschmann} one knows that there must exist some $L_\infty$ algebra structure on the chain complex 
\[
\Omega^0(X)\xrightarrow{d}\Omega^1(X)\xrightarrow{d}\cdots \xrightarrow{d}\Omega^{n-2}(X)\xrightarrow{(0,d)}{\mathrm{Ham}}^{n-1}(X)
\,
\]
 making it an $L_\infty$-algebra quasi-isomorphic to the dgla $\mathrm{dgLie}_{\mathrm{Qu}} (X,\bar{A})$. The remarkable information provided by Theorem \ref{theorem.main-theorem} is that this $L_\infty$ algebra structure is identified with that provided by Proposition \ref{ham-infty}.
\end{remark}

\begin{corollary}The image of the natural projection
  $\mathrm{dgLie}_{\mathrm{Qu}} (X,\bar{A})\to \mathfrak{X}(X)$ is the subspace
  $\mathfrak{X}_{\mathrm{Ham}}(X)$ of Hamiltonian vector
  fields. That is, the infinitesimal quantomorphisms cover infinitesimal {Hamiltonian}
  $n$-plectomorphisms.
\end{corollary}

\begin{remark}
  Theorem \ref{theorem.main-theorem} implies that
  $\mathrm{dgLie}_{\mathrm{Qu}} (X,\bar{A})$ is independent, up to equivalence, of
  the choice of prequantization $\bar{A}$ of $\omega$.
  It  also
  says that $\mathrm{dgLie}_{\mathrm{Qu}} (X,\bar{A})$ is a `rectification'
  or `semi-strictification' of the $L_\infty$-algebra
  $L_\infty(X,\omega)$.
\end{remark}

\section{Inclusion into Atiyah and Courant \texorpdfstring{$L_\infty$}{Loo}-algebras}
\label{InclusionIntoAtiyahAndCourant}
{If $(X,\omega)$ is a prequantized symplectic manifold, and $(P,A)$ is
the corresponding prequantum bundle, then there is an embedding,
induced by the morphism given in Thm.\ \ref{theorem.main-theorem}, of
the Lie algebra of observables on $X$ into the Lie algebra of
$U(1)$-invariant vector fields on $P$. The latter is the Lie algebra
of global sections of the Atiyah algebroid of $P$ (see, for example, Sec.\
2 of \cite{rogers.2-plectic} and Def.\ \ref{Atiyah_Lie_algebra}
below). The integrated analog of this embedding is a canonical map from the
group of quantomorphisms to the group of bisections \cite[Chap.\
15]{CdS-Weinstein:1999} of the Lie groupoid that integrates the
Atiyah algebroid. This groupoid is usually called the gauge
groupoid of $P$, but we prefer to call it the `Atiyah
groupoid'. Likewise, we call its group of bisections the
`Atiyah group'. Such a bisection is just an equivariant diffeomorphism of $P$
covering a diffeomorphism of the base $X$, and hence it ``forgets'' the
connection 1-form $A$.}

{In analogy with the above, we now explain how similar
embeddings of quantomorphisms naturally arise in the higher case. This
  provides the motivation for the Lie-theoretic results presented in this section.}

\subsubsection{Higher Atiyah groups and the Courant $n$-group}
Recall from Section \ref{PrequantumnPlecticManifolds}
that the $n$-stack $\mathbf{B}^nU(1)_{\mathrm{conn}}$ is presented via
the Dold-Kan correspondence by the presheaf of chain complexes
\[
C^\infty(-;U(1))\xrightarrow{d\text{log}}\Omega^1(-)\xrightarrow{d}\Omega^2(-)\xrightarrow{d}\cdots\cdots\xrightarrow{d}\Omega^{n}(-)
\]
with $\Omega^{n}(-)$ in degree zero. We can also consider 
the
$n$-stack $\mathbf{B} \bigl(\mathbf{B}^{n-1}U(1)_{\mathrm{conn}}
\bigr)$, which is the delooping of the $(n-1)$ stack $\mathbf{B}^{n-1}U(1)_{\mathrm{conn}}$.
It is presented by the presheaf
\[
C^\infty(-;U(1))\xrightarrow{d\text{log}}\Omega^1(-)\xrightarrow{d}\Omega^2(-)\xrightarrow{d}\cdots\cdots\xrightarrow{d}\Omega^{n-1}(-)
\to 0
\]
with $\Omega^{n-1}(-)$ in degree 1. In general, there is more, namely a commutative diagram
\[
\scalebox{0.92}{
\xymatrix{
C^\infty(-;U(1))\ar[r]^-{d\text{log}}\ar[d]&\Omega^1(-)\ar[d]\ar[r]^{d}&\cdots\ar[r]^-{d}&\Omega^{n-1}(-)\ar[d]\ar[r]^{d}&\Omega^{n}(-)\ar[d]\\
C^\infty(-;U(1))\ar[r]^-{d\text{log}}\ar[d]&\Omega^1(-)\ar[d]\ar[r]^{d}&\cdots\ar[r]^-{d}&\Omega^{n-1}(-)\ar[d]\ar[r]&0\ar[d]\\
\cdots\ar[d]&\cdots\ar[d]&\cdots&\cdots\ar[d]&\cdots\ar[d]\\
C^\infty(-;U(1))\ar[r]&0\ar[r]^{d}&\cdots\ar[r]&0\ar[r]&0\\
}
}
\]
corresponding to a sequence of natural forgetful morphisms of stacks
\[
\mathbf{B}^nU(1)_{\mathrm{conn}}\to \mathbf{B}(\mathbf{B}^{n-1}U(1)_{\mathrm{conn}})\to \mathbf{B}^2(\mathbf{B}^{n-2}U(1)_{\mathrm{conn}})\to \cdots \to \mathbf{B}^nU(1),
\]
where at each step the top differential form data for the connection are forgotten.

 If $\nabla:X\to \mathbf{B}^nU(1)_{\mathrm{conn}}$ is the morphism
 representing a $U(1)$-$n$-bundle with connection on a smooth manifold
 $X$, then the forgetful morphisms realize $X$ both as an object over
 $\mathbf{B}(\mathbf{B}^{n-1}U(1)_{\mathrm{conn}})$ and as an object
 over $\mathbf{B}^nU(1)$. Therefore we have a sequence of
 automorphism $n$-groups of $\nabla$:
\[
\mathbf{Aut}_{/_{\mathbf{B}^{n}U(1)_{\mathrm{conn}}}}(\nabla)\to \mathbf{Aut}_{/_{\mathbf{B}(\mathbf{B}^{n-1}U(1)_{\mathrm{conn}})}}(\nabla)\to \mathbf{Aut}_{/_{\mathbf{B}^nU(1)}}(\nabla)
\,.
\]
All of these automorphism
     $n$-groups have a smooth structure and are
   ``concretified'' in the sense of Remark
   \ref{remark.concretification}.
{We call the $n$-group $\mathbf{Aut}_{/_{\mathbf{B}^nU(1)}}(\nabla)$
the `Atiyah $n$-group' of $\nabla$, since for the case $n=1$, it is the previously
mentioned Atiyah group. We call
$\mathbf{Aut}_{/_{\mathbf{B}(\mathbf{B}^{n-1}U(1)_{\mathrm{conn}})}}(\nabla)$
the `Courant $n$-group' of $\nabla$, since for the $n=2$ case it can
be thought of as the object that integrates the Courant Lie 2-algebra.
(see Def.\ \ref{Courant_L2A} below). A more detailed discussion of these objects
as the bisections of smooth $\infty$-groupoids appears in \cite{companionArticle}.

Conceptually speaking, the infinitesimal analogue of the above sequence
of $n$-groups is a sequence of Lie $n$-algebras
\begin{equation} \label{inf_sequence}
\mathrm{Lie}\mathbf{QuantMorph}(\nabla)\to\mathrm{Lie}\mathbf{Courant}(\nabla)\to \mathrm{Lie}\mathbf{Atiyah}(\nabla),
\end{equation}
where $\mathrm{Lie}\mathbf{QuantMorph}(\nabla)$ is the Lie $n$-algebra
of infinitesimal quantomorphisms
described in the beginning of Sec.\
\ref{inf_quant_sec}.  The elements of
$\mathrm{Lie}\mathbf{Atiyah}(\nabla)$ are to be thought of as those infinitesimal
autoequivalences which preserve only the underlying $U(1)$-$n$-bundle of
$\nabla$, while $\mathrm{Lie}\mathbf{Courant}(\nabla)$ consists of those
infinitesimal autoequivalences which preserve all of the connection data on
the $n$-bundle except the highest degree part.

Recall that we modeled the Lie $n$-algebra
$\mathrm{Lie}\mathbf{QuantMorph}(\nabla)$ by using the dg Lie algebra
$\mathrm{dgLie}_{\mathrm{Qu}} (X,\bar{A})$ given in Def/Prop.\
\ref{def.dg-lie-quant}.  Similarly, we define below dg Lie algebras
which can be thought of as models for
$\mathrm{Lie}\mathbf{Atiyah}(\nabla)$ and
$\mathrm{Lie}\mathbf{Courant}(\nabla)$ for the $n=1$ and $n=2$ cases.
We consider these particular cases in order to relate our results to
the traditional theory of prequantum $U(1)$-bundles and also more recent work on Courant algebroids and $U(1)$-bundle gerbes.}

\subsection{The \texorpdfstring{$n=1$}{n1} case} \label{n.eq.1_sec}
Here $(X,\omega)$ is an ordinary pre-symplectic manifold, and the algebra
of local observables $L_{\infty}(X,\omega)$ (Def.\ \ref{TheLooXOmega})
is the underlying Lie algebra of the Poisson algebra of Hamiltonian
functions.  A prequantization $\nabla$ is an ordinary $U(1)$-principal
bundle with connection over $X$.

From any closed 2-form $\omega$, one can construct a Lie
algebroid over $X$ whose global sections form the following Lie algebra
(see for example, Sec.\ 2 of \cite{rogers.2-plectic}):
\begin{definition} \label{Atiyah_Lie_algebra}
Let $(X,\omega)$ be a presymplectic manifold. The \emph{Atiyah Lie algebra}
$\mathfrak{atiyah}(X,\omega)$ is the vector space
$\mathfrak{X}(X)  \oplus C^\infty(X;\mathbb{R})$ endowed with the Lie bracket
\[
\bbrac{v_1+c_1}{v_2+c_2}_{\mathfrak{at}}=[v_1,v_2]+\mathcal{L}_{v_1}c_2-\mathcal{L}_{v_2}c_1-\omega(v_1,v_2).
\]
\end{definition}
Obviously  the underlying vector
space of the Lie algebra $L_{\infty}(X,\omega)$
is a subspace of $\mathfrak{atiyah}(X,\omega)$. A straightforward
calculation shows that the inclusion
\begin{equation} \label{incl1}
L_{\infty}(X,\omega) \hookrightarrow \mathfrak{atiyah}(X,\omega)
\end{equation}
is also a Lie algebra morphism.
Just like in our construction of the dgla $\mathrm{dgLie}_{\mathrm{Qu}} (X,\bar{A})$
(Def/Prop.\ \ref{def.dg-lie-quant}), we now represent the prequantization
$\nabla$ by a \v{C}ech-Deligne 1-cocycle \eqref{DeligneCocycle}, and
obtain a model for $\mathrm{Lie}\mathbf{Atiyah}(\nabla)$.
\begin{definition}
If $\bar{A}=A^{1} + A^{0}$ is a {\v C}ech-Deligne
$1$-cocycle on $X$ relative to some cover $\mathcal{U}$, then
$\mathrm{Lie}_{\mathrm{At}}(X,\bar{A})$
is the Lie algebra whose underlying vector space is
\[
\mathrm{LieAlg}_{\mathrm{At}}(X,\bar{A})= \{ v +  \bar{\theta} \in \X(X) \oplus
\check{C}^0(\cU,\Omega^{0}) ~\vert ~  \L_{v} A^{0} = \delta \bar{\theta} \}
\]
with Lie bracket
$
\bbrac{v_1 + \bar{\theta}_{1}}{v_2 + \bar{\theta}_{2}}_{\mathrm{At}} = [v_1,v_2] +
\L_{v_{1}} \bar{\theta}_{2} -\L_{v_{2}} \bar{\theta}_{1}
$.
\end{definition}
Since $\L_{v} A^{0} = \iota_{v} d \log A^{0}$, it is easy to see that
an element of $\mathrm{LieAlg}_{\mathrm{At}}(X,\bar{A})$ is the local
data corresponding to a $U(1)$-invariant vector field on the
total space $P$ of the prequantum bundle, i.e., a global section of the
Atiyah algebroid $TP/U(1) \to X$. Moreover, by construction, there is an inclusion of Lie algebras
\begin{equation} \label{incl2}
  \mathrm{dgLie}_{\mathrm{Qu}} (X,\bar{A})\hookrightarrow \mathrm{LieAlg}_{\mathrm{At}}(X,\bar{A})
\end{equation}
exhibiting the infinitesimal quantomorphisms as the Lie subalgebra of
vector fields on $P$ that preserve the connection.

The following
proposition describes the relationship between
$\mathrm{LieAlg}_{\mathrm{At}}(X,\bar{A})$ and
$\mathfrak{atiyah}(X,\omega)$, which one can think of as an
extension of Thm.\ \ref{theorem.main-theorem} for the $n=1$ case.

\begin{proposition} \label{atiyah_commutes_prop}
There exists a natural Lie algebra isomorphism
$$ \psi \maps \mathfrak{atiyah}(X,\omega) \xto{\cong} \mathrm{LieAlg}_{\mathrm{At}}(X,\bar{A})$$
such that the following diagram commutes
\[
\scalebox{0.92}{\xymatrix{
\mathfrak{atiyah}(X,\omega) \ar[r]^{\psi} & \mathrm{LieAlg}_{\mathrm{At}}(X,\bar{A})\\
L_{\infty}(X,\omega) \ar[u] \ar[r]^{f} & \mathrm{dgLie}_{\mathrm{Qu}} (X,\bar{A}) \ar[u]
}}
\]
where $f \maps L_{\infty}(M,\omega) \xto{\cong}
\mathrm{dgLie}_{\mathrm{Qu}} (X,\bar{A})$ is the isomorphism of Lie algebras given in
Thm.\ \ref{theorem.main-theorem}, and
the vertical morphisms are the inclusions \eqref{incl1} and \eqref{incl2}.
\end{proposition}
\begin{proof}
It follows from Thm.\ \ref{theorem.main-theorem} that the isomorphism
$f \maps L_{\infty}(M,\omega)\! \xto{\cong}\! \mathrm{dgLie}_{\mathrm{Qu}} (X,\bar{A})$ is
$f(v + c) = v - c \vert_{U_{\alpha}} + \iota_v A^{1}$.
Hence, we define $\psi \maps \mathfrak{atiyah}(X,\omega) \to
\mathrm{LieAlg}_{\mathrm{At}}(X,\bar{A})$ to be
$\psi(v + c) = v - c \vert_{U_{\alpha}} + \iota_v A^{1}$.
Note that if $\L_{v} A^{0} = \delta \bar{\theta}$, then $\delta
(\bar{\theta} + \iota_{v}A^{1})=0$.
Hence $\psi$ is an isomorphism of vector spaces. The fact that
$\psi$ preserves the Lie brackets follows from the equalities
$
\L_{v_1} \iota_{v_2} A^{1} - \L_{v_2} \iota_{v_1} A^{1} =
\iota_{[v_1,v_2]} A^{1} + \iota_{v_{2}}\iota_{v_{1}} dA^{1} =
\iota_{[v_1,v_2]} A^{1} + \iota_{v_1 \wedge v_2} \omega.
$
\end{proof}
\begin{remark}
Note that the isomorphism $\psi$ 
in the above Proposition
uses the connection $A$
to  lift horizontally a vector field on $M$ to a vector field on $P$ in the
standard way.
\end{remark}

\subsection{The \texorpdfstring{$n=2$}{n2} case} \label{n.eq.2_sec}
Here $(X,\omega)$ is a  pre-2-plectic manifold.
A prequantization $\nabla$ of $(X,\omega)$ is a $U(1)$-bundle gerbe
(or principal $U(1)$ 2-bundle) over $X$ equipped
with a 2-connection.

In addition to the Lie 2-algebra of local observables
$L_{\infty}(X,\omega)$, there are two other Lie 2-algebras one can
build directly from any closed 3-form $\omega$. It seems that the first of these
has not appeared previously in the literature, while the
second one originates in Roytenberg and Weinstein's work on Courant
algebroids \cite{Roytenberg-Weinstein}. (The 2-term truncation we use
here is due to subsequent work by Roytenberg \cite{Roytenberg_L2A}.)

\begin{defprop}
Let $\omega$ be a pre-2-plectic structure on $X$. The \textit{Atiyah Lie 2-algebra} $\mathfrak{atiyah}(\omega)$ is the graded vector space
\[
\mathfrak{atiyah}(X,\omega)^0=\X(X);\qquad
\mathfrak{atiyah}(X,\omega)^1=\Omega^0(X);
\]
endowed with the brackets
\[
\left \llbracket \eta \right \rrbracket^{\fa}_{1}=0;  \qquad
\bbrac{v_1}{v_2}^{\fa}_2  = [v_1,v_2]; \qquad
\bbrac{v }{{\eta}}^{\fa}_2  = \L_{v} {\eta}; \qquad
\bbrac{v_1}{v_2,v_3}^{\fa}_3=-\iota_{v_1\wedge v_2\wedge v_3}\omega
\]
(with all other brackets zero by degree reasons).
\end{defprop}
\begin{defprop} \label{Courant_L2A}
Let $\omega$ be a pre-2-plectic structure on $X$. The \textit{Courant Lie 2-algebra} $\mathfrak{courant}(\omega)$ is the graded vector space
\[
\mathfrak{courant}(X,\omega)^0 =\X(X)\oplus \Omega^1(X);\qquad
\mathfrak{courant}(X,\omega)^1 =\Omega^0(X);
\]
endowed with the brackets
\[
\begin{split}
\left \llbracket \eta \right \rrbracket^{\fc}_1&=d\eta; \qquad\qquad \bbrac{v +\theta}{{\eta}}^{\fc}_2  =  \tfrac{1}{2} \iota_{v} d{\eta}  \\\\
\bbrac{v_1+\theta_1}{v_2+\theta_2}^{\fc}_2 & = [v_1,v_2]+ {\L_{v_{1}}\theta_{2}- \L_{v_{2}}\theta_{1}}  -\tfrac{1}{2} d \bigl (\iota_{v_1}\theta_{2} - \iota_{v_2}\theta_{1}
\bigr) - \iota_{v_1 \wedge v_2} \omega\\
\bbrac{v_1+\theta_1}{v_2+\theta_2, v_3+\theta_3}^{\fc}_3&=
-\tfrac{1}{6} \biggl( \innerprod{\bbrac{v_1+\theta_1}{v_2+\theta_2}^{\fc}_2}{v_3
    + \theta_3} + \mathrm{cyc.\ perm.} \biggr)
\end{split}
\]
where $\innerprod{\,}{\,}$ is the natural symmetric pairing between sections of $T^*X\oplus TX$, i.e.,
$
\innerprod{v_1 + \theta_1}{v_2 + \theta_2} := \iota_{v_1} \theta_2 + \iota_{v_2} \theta_1
$
(and with all other brackets zero by degree reasons).
\end{defprop}
The relationship between these Lie 2-algebras is given by the
next proposition.
\begin{proposition} \label{2plectic_seq_prop}
There exists a natural sequence of $L_\infty$ morphisms
\[
L_\infty(X,\omega)\xrightarrow{\phi} \mathfrak{courant}(X,\omega)\xrightarrow{\psi} \mathfrak{atiyah}(X,\omega),
\]
where the nontrivial components of the morphism $\phi$ are
\[
\begin{split}
\phi_1(v+\theta)&=v+ \theta; \qquad
\phi_1(\eta)=\eta;\qquad
\phi_2(v_1+\theta_1,v_2+\theta_2)= -\tfrac{1}{2} \left(\iota_{v_1}
  \theta_2 - \iota_{v_2} \theta_1 \right)
\end{split}
\]
and the nontrivial components of the morphism $\psi$ are
\[
\begin{split}
\psi_1(v+\theta)&=v;\qquad
\psi_1(\eta)=\eta;\qquad
\psi_2(v_1+\theta_1,v_2+\theta_2)=-\tfrac{1}{2} \left(\iota_{v_1}  \theta_2 - \iota_{v_2} \theta_1 \right)
\end{split}
\]
\end{proposition}
 \begin{proof}
The fact that $\phi$ is an $L_{\infty}$-morphism is the content of Thm.\
7.1 in \cite{rogers.2-plectic}. To show $\psi$ is a $L_{\infty}$ morphism,
we first perform several straightforward computations using the Cartan
calculus in order to obtain the following equalities:
\begin{align*}
&\scalebox{0.92}{$\psi_{2} (d \eta, v + \theta) = \psi_1
\bigl(\bbrac{\eta}{v+\theta}^{\fc} \bigr) -\bbrac{\psi_{1}(\eta)}{\psi_{1}(v+ \theta)}^{\fa}_2;$}\\
&\scalebox{0.92}{$\bbrac{v_1+\theta_1}{v_2+\theta_2, v_3+\theta_3}^{\fc}_3=
-\frac{1}{4} \bigl(\iota_{v_3} \L_{v_{1}} \theta_{2} -
\iota_{v_3} \L_{v_{2}} \theta_{1} + \mathrm{cyc.\ perm.} \bigr) +
\frac{1}{2} \iota_{v_1\wedge v_2\wedge v_3}\omega;$} \\
&\scalebox{0.92}{$\psi_{2}\bigl( \bbrac{v_1+\theta_1}{v_2+\theta_2}^{\fc}_2, v_3 +
\theta_3 \bigr) + \mathrm{cyc.\ perm.} = -\frac{1}{4} \bigl(\iota_{v_3} \L_{v_{1}} \theta_{2} -
\iota_{v_3} \L_{v_{2}} \theta_{1} + \mathrm{cyc.\ perm.} \bigr)$} \\
&\scalebox{0.92}{$ \quad\phantom{mmmmmmmmmmmmmmmmmmmmmmi} - \bigl( \iota_{v_1 \wedge v_2} d\theta_{3} + \mathrm{cyc.\
  perm.} \bigr) - \frac{3}{2} \iota_{v_1\wedge v_2\wedge v_3}\omega; $}\\
&\scalebox{0.92}{$\bbrac{\psi_{1}(v_{1} + \theta_1)}{\psi_{2}(v_2 + \theta_2,v_3 +
  \theta_3)}^{\fa}_{2} + \mathrm{cyc.\ perm.} = -\frac{1}{2}\bigl(\iota_{v_3} \L_{v_{1}} \theta_{2} -
\iota_{v_3} \L_{v_{2}} \theta_{1} + \mathrm{cyc.\ perm.} \bigr)$} \\
&\scalebox{0.92}{$ \quad\phantom{mmmmmmmmmmmmmmmmmmmmmmmmi} - \bigl( \iota_{v_1 \wedge v_2} d\theta_{3} + \mathrm{cyc.\  perm.} \bigr)$}.
\end{align*}
We then use the above to verify that the equalities given in \cite[Def.\ 34]{HDA6} are
satisfied.
\end{proof}

If $(X,\omega)$ is prequantized, then we represent the prequantum
2-bundle $\nabla:X\to \mathbf{B}^2U(1)_{\mathrm{conn}}$
with a \v{C}ech-Deligne 2-cocycle, and obtain
dg Lie algebras that we think of as modeling the previously discussed
$L_{\infty}$-algebras $\mathrm{Lie}\mathbf{Atiyah}(\nabla)$ and
$\mathrm{Lie}\mathbf{Cour}\text{-}\allowbreak\mathbf{ant}(\nabla)$.
In what follows, $\Omega^{\leq 0}$ and
$\Omega^{\leq 1}$ denote the cochain complexes of sheaves
\[
\Omega^0(-)\to 0\to 0\to0\to\cdots; \qquad
\Omega^0(-)\xrightarrow{d}\Omega^1(-)\to 0\to 0\to\cdots,
\]
respectively, with both having $\Omega^0(-)$ in degree zero.
\begin{defprop}
If $\bar{A}=A^{2} + A^{1} + A^{0}$ is a {\v C}ech-Deligne
$2$-cocycle on $X$ relative to some cover $\mathcal{U}$,
then we denote by $\mathrm{dgLie}_{\mathrm{At}}(X,\bar{A}) $ and
$ \mathrm{dgLie}_{\mathrm{Cou}}(X,\bar{A}) $
the dg-Lie algebras whose underlying complexes are
\begin{align*}
\mathrm{dgLie}_{\mathrm{At}}(X,\bar{A})^{0} &= \{ v +  \bar{\theta} \in \X(X) \oplus
\Tot^{1}(\cU,\Omega^{ \leq 0}) ~\vert ~  \L_{v} A^{0} = \dt \bar{\theta} \}\\
\mathrm{dgLie}_{\mathrm{At}}(X,\bar{A})^{1} &= \Tot^{0}(\cU,\Omega^{ \leq 0})
\end{align*}
and
\begin{align*}
\mathrm{dgLie}_{\mathrm{Cou}}(X,\bar{A})^{0} &= \{ v +  \bar{\theta} \in \X(X) \oplus
\Tot^{1}(\cU,\Omega ^{ \leq 1}) ~\vert ~  \L_{v}(A^{1} + A^{0}) = \dt \bar{\theta} \}\\
\mathrm{dgLie}_{\mathrm{Cou}}(X,\bar{A})^{1} &= \Tot^{0}(\cU,\Omega^{ \leq 1});
\end{align*}
both equipped with the differential $0 \oplus \dt$,
and whose graded Lie brackets are (for both cases)
\begin{equation}
\begin{split}
&\bbrac{v_1 + \bar{\theta}_{1}}{v_2 + \bar{\theta}_{2}}  = [v_1,v_2] +
\L_{v_{1}} \bar{\theta}_{2} -\L_{v_{2}} \bar{\theta}_{1}\\
&\bbrac{v + \bar{\theta}}{\bar{\eta}}  = -\bbrac{\bar{\eta}}{v + \bar{\theta}}= \L_{v} \bar{\eta};
\qquad \qquad
\bbrac{\bar{\eta}}{\bar{\eta}}=0.
\end{split}
\end{equation}
 \label{dgLieXOmegaCou}
\end{defprop}
The dg Lie algebra $\mathrm{dgLie}_{\mathrm{At}}(X,\bar{A})$ was
constructed by Collier \cite[Def.\ 6.11, Thm.\ 8.18]{Collier:2011},
and he rigorously proved that its degree zero elements correspond to
infinitesimal autoequivalences of the $U(1)$ 2-bundle represented by
the \v{C}ech 2-cocycle $A^0$.  He also constructed
$\mathrm{dgLie}_{\mathrm{Cou}}(X,\bar{A})$ and proved that
its degree zero elements are the infinitesimal autoequivalences the $U(1)$
2-bundle equipped with a connective structure represented by the
truncated \v{C}ech-Deligne 2-cocycle $A^{1} + A^{0}$ \cite[Def.\
10.38, Prop.\ 10.48]{Collier:2011}.
There is an obvious map of dg Lie algebras
$\mathrm{dgLie}_{\mathrm{Cou}}(X,\bar{A}) \xto{p}
\mathrm{dgLie}_{\mathrm{At}}(X,\bar{A})$, which in degree zero forgets the $\check{C}^{0}(\cU,\Omega^1)$ component.
It is also clear that the dg Lie algebra $\mathrm{dgLie}_{\mathrm{Qu}}(X,\bar{A})$ of
infinitesimal quantomorphisms (Def/Prop.\ \ref{def.dg-lie-quant})
embeds into $\mathrm{dgLie}_{\mathrm{Cou}}(X,\bar{A})$.
Hence the next result follows automatically by construction.
\begin{proposition} \label{dgla_incl_prop}
There is a natural sequence of dg Lie algebras
\[
\mathrm{dgLie}_{\mathrm{Qu}} (X,\bar{A})\xto{i}\mathrm{dgLie}_{\mathrm{Cou}}(X,\bar{A})
\xto{p} \mathrm{dgLie}_{\mathrm{At}}(X,\bar{A})
\]
that we interpret as modeling the sequence \eqref{inf_sequence}.
\end{proposition}

In \cite[Theorem 12.50]{Collier:2011}, Collier constructed a weak
equivalence of Lie 2-algebras between a local \v{C}ech description of
the Courant Lie 2-algebra \eqref{Courant_L2A} and the dg Lie algebra
$\mathrm{dgLie}_{\mathrm{Cou}}(X,\bar{A})$. We conclude with the
following proposition which strengthens this result by incorporating
our Thm.\ \ref{theorem.main-theorem} and Prop.\
\ref{2plectic_seq_prop}. It can also be viewed as the higher analog of
Prop.\ \ref{atiyah_commutes_prop}.
\begin{proposition}
If $(X,\omega)$ is a prequantized pre-2-plectic manifold, then there
exist natural weak equivalences of Lie 2-algebras
$
f^{\fa} \maps \mathfrak{atiyah}(X,\omega) \xto{\sim} \mathrm{dgLie}_{\mathrm{At}}(X,\bar{A})
$
and
$
f^{\fc} \maps \mathfrak{courant}(X,\omega) \xto{\sim} \mathrm{dgLie}_{\mathrm{Cou}}(X,\bar{A})
$
such that the following diagram of $L_\infty$-algebras (strictly) commutes:
\[
\scalebox{0.92}{
\xymatrix{
\mathfrak{atiyah}(X,\omega) \ar[r]^{f^{\fa}} & \mathrm{dgLie}_{\mathrm{At}}(X,\bar{A}) \\
\mathfrak{courant}(X,\omega) \ar[u]^{\psi} \ar[r]^{f^{\fc}} & \ar[u]_{p} \mathrm{dgLie}_{\mathrm{Cou}}(X,\bar{A}) \\
L_{\infty}(X,\omega) \ar[u]^{\phi} \ar[r]^{f} & \ar[u]_{i} \mathrm{dgLie}_{\mathrm{Qu}} (X,\bar{A})\\
}}
\]
where $f \maps L_{\infty}(M,\omega) \xto{\sim} \mathrm{dgLie}_{\mathrm{Qu}} (X,\bar{A})$
is the weak equivalence given in Thm.\ \ref{theorem.main-theorem}, and
the vertical morphisms are those given in Prop.\
\ref{2plectic_seq_prop} and Prop.\ \ref{dgla_incl_prop}.
\end{proposition}
\begin{proof}
In terms of the notation above, Prop.\
\ref{f1_prop} and Eqs.\ (\ref{struct_maps}) imply that the weak
equivalence $f \maps L_{\infty}(M,\omega) \xto{\sim}
\mathrm{dgLie}_{\mathrm{Qu}} (X,\bar{A})$ has non-trivial components
\begin{align*}
f_{1}( v + \theta) & = v - \theta + \iota_v(A^2 - A^1); \qquad \qquad
f_{1}(\eta) = -\eta;\\
f_{2}(v_1 + \theta_1, v_2 + \theta_2) &= \iota_{v_{1}} \theta_2  -
\iota_{v_{2}} \theta_1 + \iota_{v_1 \wedge v_2} A^2.
\end{align*}
(Above we have suppressed the restriction of global forms on $X$ to
open sets $U_\alpha \in \mathcal{U}$.)
Hence, we define the non-trivial components of $f^{\fc}$ to be
\begin{align*}
f^{\fc}_{1}( v + \theta) & = v - \theta + \iota_v(A^2 - A^1); \qquad\qquad
f^{\fc}_{1}(\eta) = -\eta;\\
f^{\fc}_{2}(v_1 + \theta_1, v_2 + \theta_2) &= \tfrac{1}{2} \bigl(\iota_{v_{1}} \theta_2  -
\iota_{v_{2}} \theta_1 \bigr) + \iota_{v_1 \wedge v_2} A^2.
\end{align*}
Similarly, we define $f^{\fa}$ by
\begin{align*}
f^{\fa}_{1}( v) & = v - \iota_v A^1; \qquad
f^{\fa}_{1}(\eta) = -\eta;\qquad
f^{\fa}_{2}(v_1, v_2) =\iota_{v_1 \wedge v_2} A^2.
\end{align*}
Note that if $v + \bar{\theta}$ is a degree 0 element of
$\mathrm{dgLie}_{\mathrm{Cou}}(X,\bar{A})$, then $\dt ( \bar{\theta} -
\iota_{v}(A^2 - A^{1})) =0$. Similarly, if
$v + \bar{\theta}$ is a degree 0 element of
$\mathrm{dgLie}_{\mathrm{At}}(X,\bar{A})$, then $\dt ( \bar{\theta} +
\iota_{v}A^{1}) =0$. It then follows from the Poincar\'{e} Lemma
that both $f^{\fc}_{1}$ and $f^{\fa}_{1}$ are quasi-isomorphisms of
chain complexes.

It follows immediately from the definitions (Prop.\
\ref{2plectic_seq_prop} and Prop.\ \ref{dgla_incl_prop}) that
$f^{\fc}_{1} \circ \phi_1 = i \circ f_{1}$ and $f^{\fa}_{1} \circ
\psi_1 = p \circ f^{\fc}_{1}$. Simple calculations show that the
following equations hold:
\begin{align*}
(f^{\fc} \circ \phi)_{2}(v_1 + \theta_1,v_2 + \theta_2) &:=
f^{\fc}_{1} \phi_{2}(v_1 + \theta_1,v_2 + \theta_2) +
f^{\fc}_{2}(\phi_1(v_1 + \theta_1),\phi_1(v_2 + \theta_2)) \\
& ~ = i \circ f_{2}(v_1 + \theta_1,v_2 + \theta_2),\\
(f^{\fa} \circ \psi)_{2}(v_1 + \theta_1,v_2 + \theta_2) &:=
f^{\fa}_{1} \psi_{2}(v_1 + \theta_1,v_2 + \theta_2) +
f^{\fa}_{2}(\psi_1(v_1 + \theta_1),\psi_1(v_2 + \theta_2)) \\
& ~ = p \circ f^{\fc}_{2}(v_1 + \theta_1,v_2 + \theta_2).
\end{align*}
Hence the above diagram commutes.
Next, using the identities from Sec.\ \ref{notat_cartan} and
the cocycle equation for $\bar{A}$, we obtain the following
equalities:
\begin{align*}
&\scalebox{0.92}{$\bbrac{f^{\fc}_{1}(v_1 + \theta_1)}{f^{\fc}_{1}(v_2 +
  \theta_2)}^{\mathrm{Cou}}= [v_1,v_2] - \L_{v_1} \theta_{2} +
\L_{v_2} \theta_1 + \iota_{[v_1,v_2]}(A^{2}-A^{1})$} \\
&\scalebox{0.92}{$\quad\phantom{mmmmmmmmmmmmmmmmmmmmmmmmmmmmmm}
 + \iota_{v_1 \wedge  v_2} \omega - \dt (\iota_{v_1 \wedge v_2} A^{2})$}\\
&\scalebox{0.92}{$f^{\fc}_{2}\bigl( \bbrac{v_1+\theta_1}{v_2+\theta_2}^{\fc}_2, v_3 +
\theta_3 \bigr) + \mathrm{cyc.\ perm.} = \frac{1}{4} \bigl(\iota_{v_3} \L_{v_{1}} \theta_{2} -
\iota_{v_3} \L_{v_{2}} \theta_{1} + \mathrm{cyc.\ perm.} \bigr) $}\\
& \scalebox{0.92}{$\quad \phantom{mmmmmmmmmmmmmmmm}
+ \bigl( \iota_{v_1 \wedge v_2} d\theta_{3} + \iota_{[v_1,v_2]
  \wedge v_3}A^{2} + \mathrm{cyc.\
  perm.} \bigr) + \frac{3}{2} \iota_{v_1\wedge v_2\wedge v_3}\omega$} \\
&\scalebox{0.92}{$\bbrac{f^{\fc}_{1}(v_{1} + \theta_1)}{f^{\fc}_{2}(v_2 + \theta_2,v_3 +
  \theta_3)}^{\mathrm{Cou}}_{2} + \mathrm{cyc.\ perm.} =
 \frac{1}{2}\bigl(\iota_{v_3} \L_{v_{1}} \theta_{2} -
\iota_{v_3} \L_{v_{2}} \theta_{1} + \mathrm{cyc.\ perm.} \bigr)$} \\
 &\scalebox{0.92}{$\quad \phantom{mmmmmmmmmmmmmmm}
 + \bigl( \iota_{v_1 \wedge v_2} d\theta_{3} + 2 \iota_{[v_1,v_2]
  \wedge v_3}A^{2} + \iota_{v_2 \wedge v_3} \L_{v_1}A^{2} + \mathrm{cyc.\   perm.} \bigr)$}.
\end{align*}
And similarly for $f^{\fa}$:
\begin{align*}
\bbrac{f^{\fa}_{1}(v_1)}{f^{\fa}_{1}(v_2)}^{\mathrm{At}} & = [v_1,v_2]
- \iota_{[v_1,v_2]} A^{1} - \dt (\iota_{v_1 \wedge v_2} A^{2}) \\
f^{\fa}_{2}\bigl (\bbrac{v_1}{v_2}^{\fa}_{2},v_3 \bigr)  + \mathrm{cyc.\
  perm.} &= \iota_{[v_1,v_2] \wedge v_3} A^{2} + \mathrm{cyc.\
  perm.}\\
\bbrac{f^{\fa}_{1}(v_1)}{f^{\fa}_{2}(v_2,v_3)}^{\mathrm{At}} + \mathrm{cyc.\
  perm.} &= \bigl(2 \iota_{[v_1,v_2] \wedge v_3} + \iota_{v_2 \wedge
  v_3} \L_{v_{1}} \bigr) A^{2} + \mathrm{cyc.\  perm.}.
\end{align*}
Using these in conjunction with Lemma \ref{tech_lemma}, it is  easy
to verify that $f^{\fc}$ and $f^{\fa}$ are $L_{\infty}$ morphisms (see, e.g.,
\cite[Def.\ 34]{HDA6}).
\end{proof}

\appendix

\section{An explicit weak equivalence between
  \texorpdfstring{$L_{\infty}(X,\omega)$}{LooXo} and \texorpdfstring{$\LieQuant$}{LieQuant}} \label{appendix.proof}

\numberwithin{equation}{section}

In this section, we prove Thm.\ \ref{theorem.main-theorem}. Namely,
given a pre $n$-plectic manifold $(X,\omega)$ and a prequantization
presented by a \v{C}ech-Deligne $n$-cocycle $\bar{A}$ with respect to
a cover $\cU= \{U_{\alpha} \}$ of $X$, we shall construct
an
$L_{\infty}$-quasi-isomorphism
$
L_{\infty}(X,\omega) \xto{\sim} \LieQuant.
$
We use the following conventions to help simplify calculations:
\begin{itemize}

\item{
We denote by
$\res \maps \Omega^{\bullet}(X) \to \check{C}^{0}(\cU,\Omega^{\bullet})$
the restriction map
$
\res(\theta)_{\alpha} = \theta \vert_{U_{\alpha}} \in \Omega^{\bullet}(U_{\alpha})
$.
For $\bar{A} = \sum_{i=0}^{n}A^{n-i}$ a \v{C}ech-Deligne cocycle, we
define for all $m \geq 1$:
\begin{equation} \label{Asign}
\bar{A}(m) := \sum_{i=0}^{n} (-1)^{mi} A^{n-i}.
\end{equation}

}

\item{
    $(L,l_1)$ denotes the underlying complex of the Lie $n$-algebra
    $L_{\infty}(X,\omega)$ introduced in Def.\ \ref{def.local-observables}:
    \begin{equation*}
      \begin{split}
        L_{0} & = \{ v +  H \in \Xham(X) \oplus \Omega^{n-1}(X)  ~ \vert ~ dH = -\iota_{v}\omega \} \\
        L_{i} & = \Omega^{n-1-i}(X) \quad i \geq 1
      \end{split}
    \end{equation*}
    with differential
    \begin{equation*}
      l_1 \theta = \begin{cases} 0 + d \theta, & \deg{\theta} =1\\    d
        \theta, & \deg{\theta}  > 1. \end{cases}
    \end{equation*}
The higher $k$-ary brackets of $L_{\infty}(X,\omega)$ are denoted by $l_{2},\ldots,l_{n+1}$.
}
\item{
   $(L',l'_1)$ denotes the underlying complex
    of the dgla $\LieQuant$ introduced in Def.\ \ref{def.dg-lie-quant}:
\begin{equation*}
\begin{split}
L'_{0} &= \{ v +  \bar{\theta} \in \X(X) \oplus
\Tot^{n-1}(\cU,\Omega) ~\vert ~  \L_{v} \bar{A} = \dt \bar{\theta} \}\\
L'_{i} &= \Tot^{n-1 - i}(\cU,\Omega) \quad
i \geq 1,
\end{split}
\end{equation*}
with differential
    \begin{equation*}
      l'_1 \bar{\theta} = \begin{cases} 0 + \dt \bar{\theta}, & \deg{\bar{\theta}} =1\\    \dt
        \bar{\theta}, & \deg{\bar{\theta}}  > 1. \end{cases}
    \end{equation*}
The Lie bracket on $\LieQuant$ is denoted by $l'_{2}= \bbrac{\cdot}{\cdot}$.
}
\item{
    Elements of arbitrary degree in $L$ (resp.\ $L'$) will be denoted
    as $x_1,x_2,\ldots$ (resp.\ $\bx_1,\bx_2,\ldots$) where
    \begin{equation}\label{notation_app}
      x_i := v_i + \theta_i \quad \text{(resp.\ $\bx_i := v_i + \btheta_i$).}
    \end{equation}
    It is understood that we set $v_i=0$ if
    $\deg{x_{i}} > 0$ (resp.\ $\deg{\bx_{i}} > 0$). So for example, for any $x_{1},\ldots,x_k \in L$ and
    any $\bx_{1}, \bx_2 \in L'$ the following equalities hold:
\begin{equation*} \label{brackets_app}
\begin{split}
l_2(x_1,x_2) &= [v_1,v_2] + \iota_{v_1 \wedge v_2}\omega;\qquad
l_{k \geq 3}(x_1,\ldots,x_k)  = \ksgn{k} \iota_{v_1 \wedge \cdots \wedge v_k}\omega,\\
\bbrac{\bx_1}{\bx_2} & = [v_1, v_2] + \L_{v_1} \btheta_2 - \L_{v_{2}} \btheta_1.
\end{split}
\end{equation*}
}

\item{
For all $m \geq 2$ we define a map
$
S_{m} \maps L^{\tensor m} \to L,
$
where
\begin{equation} \label{Sm_def}
S_{m}(x_{1},\ldots,x_{m})  =\sum_{i=1}^{m}(-1)^{i} \iota_{v_{1} \wedge \cdots \wedge
\widehat{v}_{i} \wedge \cdots \wedge v_{m}} \theta_{i}.
\end{equation}
It is clear from our above notation that $S_{m}(x_{1},\ldots,x_{m})=0$
if two or more arguments have degree $>0$.
Note that $S_{m}$ is a graded skew-symmetric map
of degree $m-1$ and $S_{m > n}=0$.
}

\item{
For all $m \geq 1$ we define the linear maps $f_{m} \maps L^{\tensor
  m} \to L'$:
\begin{equation} \label{f1_map}
f_1(x) = v - \res(\theta) + \iota_{v} \bar{A}(1),
\end{equation}
 \begin{equation}\label{struct_maps}
f_{ 2\leq m \leq n}(x_1,\ldots,x_m) = \ksgn{m} \left( \res \circ S_{m}(x_1,\ldots,x_{m}) +  \iota_{v_1 \wedge \cdots \wedge v_m}\bar{A}(m) \right),
\end{equation}
and $f_{m > n} =0.$ Note that each $f_m$ is graded skew-symmetric
with $\deg{f_m}=m-1$. Below, we will often suppress the restriction map in the definitions.
These are the structure maps we will use to construct an $L_{\infty}$ quasi-morphism.
}

\item{
Finally,  we
define the following auxiliary linear maps $I^{m}_{(1)}, ~ I^{m}_{(2)}, ~ I^{m}_{(3)}  \maps L^{\tensor m} \to L'$, for all $m \geq 1$,
where $I^{1}_{2}=I^{m <3}_{3}=0$ and
\begin{equation} \label{I_maps}
\begin{split}
\scalebox{0.86}{$I^{m}_{(1)}(x_{1},\ldots,x_{m})$} &\scalebox{0.86}{$= \displaystyle{\sum_{ \sigma \in \Sh(1,m-1)}} \chi(\sigma) (-1)^{m} f_{m}(l_{1}(x_{\sigma(1)}),\ldots,x_{\sigma(m)}),$} \\
\scalebox{0.86}{$I^{m \geq 2}_{(2)}(x_{1},\ldots,x_{m}) $}&\scalebox{0.86}{$=-\displaystyle{\sum_{ \sigma \in \Sh(2,m-2)}} \chi(\sigma) f_{m-1}(l_{2}(x_{\sigma(1)},x_{\sigma(2)}),\ldots,x_{\sigma(m)}),$} \\
\scalebox{0.86}{$ I^{m \geq 3}_{(3)}(x_{1},\ldots,x_{m})$} &\scalebox{0.86}{$= \displaystyle{\sum_{\substack{k=3\dots m \\ \sigma \in  \Sh(k,m-k)}}} \chi(\sigma)  (-1)^{k(m-k) +1}
f_{m+1 -k}(l_{k}(x_{\sigma(1)},\hdots,x_{\sigma(k)}),\ldots,x_{\sigma(m)}).$}
\end{split}
\end{equation}
Above, $\chi(\sigma)=(-1)^{\sigma} \epsilon(\sigma)$, where
$\epsilon(\sigma)$ is the Koszul sign of the permutation.
We also define for all $m \geq 1$ maps
$J^{m} \maps L^{\tensor m} \to L'$, where $J^{1} =0$ and
\begin{equation} \label{J_maps}
\begin{split}
& \scalebox{0.87}{$J^{m \geq 2} (x_1,\hdots,x_m)=$} \\
&  \scalebox{0.87}{$ \displaystyle{\sum_{\substack{s+t = m\\ \tau \in \Sh(s,m-s) \\ \tau(1) <
    \tau(s+1)}}} \chi(\tau) (-1)^{s-1} (-1)^{(t-1) \sum_{p=1}^{s}
  \deg{x_{\tau(p)}}}
\bbrac{f_{s}(x_{\tau(1)},\hdots,x_{\tau(s)})}{f_{t}(x_{\tau(s+1)},\hdots,x_{\tau(m)})}.$}
\end{split}
\end{equation}
}
\end{itemize}

\begin{propositionapp}\label{f1_prop}
The linear map \eqref{f1_map} $f_{1} \maps L \to L'$
is a quasi-isomorphism of chain complexes.
\end{propositionapp}

\begin{proof}
It is clear from the definition that $f_1$ is a chain map.
Since $\bar{A}$ is a \v{C}ech-Deligne cocycle, and since the interior
product $\iota_v$ commutes with the \v{C}ech differential, we have
$
\dt \iota_v \bar{A}(1) = d\iota_v A^{n} + \sum_{i=1}^{n} \L_{v} A^{n-i}.
$
This implies that $v + \bar{\theta} \in L'_{0}$ if and only if
$\dt \bigl(\bar{\theta} -\iota_{v} \bar{A}(1) \bigr) = \res(\iota_v \omega).$
Let $\tilde{L}$ be the complex whose underlying graded vector space is
\begin{equation*}
\tilde{L}_{0} 
= \{ v +  \btheta \in \X(X) \oplus
\Tot^{n-1}(\cU,\Omega) ~\vert ~  \dt \btheta = \res(\iota_v\omega) \}
;\qquad\qquad
\tilde{L}_{i} 
 = L'_{i} \quad i > 0,
\end{equation*}
and whose differential is $\tilde{l}_1=l'_{1}$, the same differential as on $L'$. The chain map
$f_1$ then is equal to the composition:
$
L \xto{r} \tilde{L} \xto{\phi} L',
$
where, using notation \eqref{notation_app}, $r(x)= v - \res(\theta)$,
and $\phi(\bar{x})=v + \btheta + \iota_v \bar{A}(1)$.
Note that $\phi$ is a isomorphism of complexes.

Next, let $\{\rho_{\alpha} \}$ be a partition of unity subordinate to the
cover $\cU = \{U_{\alpha} \}$. Define a map $K \maps \check{C}^{i}(\mathcal{U},\Omega^{j}) \to \check{C}^{i-1}(\mathcal{U},\Omega^{j})$
to be
$
(K \theta)_{\alpha_0,\ldots,\alpha_{i-1}} = \sum_{\alpha}
\rho_{\alpha} \theta_{\alpha,\alpha_0,\ldots,\alpha_{i-1}},
$
and let $D'' \maps \check{C}^{i}(\mathcal{U},\Omega^{j}) \to \check{C}^{i}(\mathcal{U},\Omega^{j+1})$
be the ``signed'' de Rham differential $D'' \theta = (-1)^{i} d
\theta$. Then, see \cite[Prop.\ 9.5]{Bott-Tu}, there exists a chain map $
\jmath \maps \Tot^{\bullet}(\cU,\Omega) \to \Omega^{\bullet}(X)$
such that
\begin{equation}\label{collate_eq}
\jmath \circ \res = \id_{\Omega^{\bullet}(X)}, \quad
\id_{\Tot^{\bullet}(\cU,\Omega)} - \res \circ \jmath = \dt H + H \dt,
\end{equation}
where $H \maps \Tot^{\bullet}(\cU,\Omega) \to \Tot^{\bullet}(\cU,\Omega)$ is the chain homotopy given as follows:
if $\bar{\theta} = \sum_{i=0}^{m} \theta^{m-i}$,
with  $\theta^{m-i} \in \check{C}^{i}(\cU,\Omega^{m-i})$, then
$
H(\bar{\theta}) = \sum_{i=0}^{m-1} (H\bar{\theta})_{i},
$
where
\begin{equation} \label{homotopy_eq}
\begin{split}
(H\bar{\theta})_{i} &= \sum_{j=i+1}^{m} K \circ
\underset{j-(i+1)}{\underbrace{(-D'' K) \circ (-D'' K) \circ \cdots
    \circ (-D'' K)}} \theta^{m-j} \in \check{C}^{i}(\cU,\Omega^{m-1-i}).
\end{split}
\end{equation}
Hence, the restriction map $\res$ is a quasi-isomorphism between the
de Rham and \v{C}ech-de Rham complexes, and $\jmath$ is its  homotopy inverse.

Let $ \tilde{\jmath} \maps \tilde{L} \to L$ to be the chain map
$\tilde{\jmath}(\bar{x})= v - \jmath(\btheta)$.
Note that $\tilde{\jmath}$ is well-defined on degree 0 elements since
$\dt \bar{\theta} = \res(\iota_v \omega)$.
Let $\tilde{H} \maps \tilde{L} \to \tilde{L}$ be the degree 1 map
$\tilde{H}(\bar{x}) =H(\btheta)$.
We now show that $\tilde{H}$ is a chain homotopy i.e.,
$\id_{\tilde{L}} - r \circ \tilde{\jmath} = \tilde{l}_1 \tilde{H} + \tilde{H} \tilde{l_{1}}.
$ 
Since \eqref{collate_eq} holds, it follows that we just need to check 
this
on degree 0 elements.
Since we have the equality $\dt \bar{\theta} = \res(\iota_v \omega) \in \check{C}^{0}(\cU,\Omega^{n})$
for all $v + \bar{\theta} \in \tilde{L}_{0}$, it follows from the definition of $H$ \eqref{homotopy_eq} that
$H \bigl(\dt \bar{\eta} \bigr) =0$. So \eqref{collate_eq} implies that 
the above identity
holds
for degree 0 as well. Therefore $r$ is a quasi-isomorphism, and hence $f_1$ is a quasi-isomorphism.
\end{proof}

\subsection*{Technical lemmas}
In the remainder of the appendix we show that
the maps $f_{2 \leq m \leq n} \maps L^{\tensor m} \to L'$ given by equation \eqref{struct_maps} lift the  map $f_1
\maps L \to L'$ to an $L_{\infty}$-morphism between
$L_{\infty}(X,\omega)$ and $\LieQuant$. Prop.\ \ref{f1_prop} implies
that this lift will be an $L_{\infty}$-quasi-isomorphism.
We present here several small computational results necessary for the proof.
\begin{lemmaapp} \label{I1_lemma}
For all $m \geq 2$ and  $x_{1}, \ldots,x_{m} \in L$, we have
\begin{equation} \label{I1_lemma_eq0}
\begin{split}
I^{ m\geq 2}_{(1)} (x_1,\hdots,x_m) &= \ksgn{m} (-1)^{m} \sum_{i=1}^{m}
(-1)^{i}\iota_{v_{1} \wedge \cdots \wedge \widehat{v}_{i} \wedge \cdots \wedge
v_{m}} (l_{1} \theta_{i}).
\end{split}
\end{equation}
\end{lemmaapp}

\begin{proof}
Equations \eqref{struct_maps} and \eqref{I_maps} imply that
\begin{equation}\label{I1_lemma_eq1}
\begin{split}
\scalebox{0.995}{$I^{ m\geq 2}_{(1)} (x_1,\hdots,x_m)$}&\scalebox{0.995}{$= \ksgn{m}(-1)^{m}\displaystyle{ \sum_{i=1}^{m}}
(-1)^{i-1} \epsilon(\sigma(i)) S_{m}( l_{1}x_{i},x_1,\ldots,\widehat{x}_{i},\ldots,x_{m}).$}
\end{split}
\end{equation}
The vector field associated to $l_{1}x_{i}=l_{1}(v_{i}+ \theta_i)$ is zero, hence
\[
S_{m}(
l_{1}x_{i},x_1,\ldots,\widehat{x}_{i},\ldots,x_{m})=-\iota_{v_{1}
\wedge \cdots \wedge \widehat{v}_{i} \wedge \cdots \wedge v_{m}} (l_{1} \theta_{i}),
\]
and furthermore, any non-zero terms contributing to the sum \eqref{I1_lemma_eq1}
necessarily have $\epsilon(\sigma(i))=1$.

\end{proof}

\begin{lemmaapp}\label{I2_lemma}
If $x_1,x_2 \in L$, then
$
I^{2}_{(2)}(x_1,x_2) = -[v_1,v_2] + \iota_{v_1 \wedge v_2} \omega
-\iota_{[v_1,v_2]} \bar{A}(1),
$ 
and for all $m > 2$ and  $x_{1}, \ldots,x_{m} \in L$, the following
equality holds:
\begin{equation}
\begin{split}
\scalebox{0.86}{$I^{m > 2}_{(2)}(x_1,\ldots,x_m)$} &\scalebox{0.86}{$= - (-1)^{\binom{m}{2}} \left( \binom{m}{2}
  \iota_{v_1 \wedge \cdots \wedge v_m} \omega + \displaystyle{\sum_{i
    <k}} (-1)^{i+k} \iota_{[v_i,v_k] \wedge v_{1} \wedge \cdots
    \widehat{v}_{i} \cdots \widehat{v}_{k} \cdots \wedge
    v_{m}}\bar{A}(m-1) \right)$} \\
& \scalebox{0.86}{$+ (-1)^{\binom{m}{2}} \Bigl(  \displaystyle{\sum_{i < k < j} - \sum_{i
      < j < k} + \sum_{j < i < k}} \Bigr) (-1)^{i+k+j}
    \iota_{[v_i,v_k] \wedge v_1 \cdots \widehat{v}_{i} \cdots \widehat{v}_{j}
  \cdots \widehat{v}_{k} \cdots \wedge v_m} \theta_{j}$}
\end{split}
\end{equation}
\end{lemmaapp}
\begin{proof}
The $m=2$ case follows immediately from the definitions.
For  $m >2$, recall the definition of $I^{m}_{(2)}$ \eqref{I_maps}
and note the following equality of summations:
\begin{equation}\label{I2_lemma_eq0}
-\sum_{ \sigma
  \in \Sh(2,m-2)} \chi(\sigma) = \sum_{1 \leq i < k \leq m}(-1)^{i+k} \epsilon(i,k).
\end{equation}
A summand contributing to  $I^{m}_{(2)}$ is of the form
\begin{equation} \label{I2_lemma_eq1}
\begin{split}
\scalebox{0.86}{$f_{m-1}(l_{2}(x_{i},x_{k}),x_1,\ldots,
  \widehat{x}_i,\ldots,\widehat{x}_k,\ldots x_{m})$} & \scalebox{0.86}{$= - (-1)^{\binom{m}{2}}  \Bigl(
  S_{m-1}(l_{2}(x_{i},x_{k}),x_1,\ldots,\widehat{x}_i,\ldots,\widehat{x}_k,\ldots
  x_{m})$} \\
&\scalebox{0.86}{$\quad + \iota_{[v_i,v_k] \wedge v_{1} \wedge \cdots \widehat{v}_{i}
  \cdots \widehat{v}_{k} \cdots \wedge v_{m}}\bar{A}(m-1) \Bigr).$}
\end{split}
\end{equation}
The second term on the right-hand side above vanishes if
$\deg{x_{i}}>0$ for any $i$, hence taking the summation
\eqref{I2_lemma_eq0} of all such terms gives
\begin{equation} \label{I2_lemma_eq1.5}
\sum_{1 \leq i < k \leq m}(-1)^{i+k} \iota_{[v_i,v_k] \wedge v_{1} \wedge \cdots \widehat{v}_{i}
  \cdots \widehat{v}_{k} \cdots \wedge v_{m}}\bar{A}(m-1).
\end{equation}

Using \eqref{Sm_def}, we rewrite the first term on the right-hand side
of \eqref{I2_lemma_eq1} as
\begin{equation}\label{I2_lemma_eq2}
\begin{split}
  S_{m-1}(l_{2}(x_{i},x_{k}),x_1,&\ldots,\widehat{x}_i,\ldots,\widehat{x}_k,\ldots
  x_{m}) = (-1)^{i+k} \iota_{\vk{m}}\omega\\
& + \Bigl( -\sum_{j=1}^{i} + \sum_{j=i+1}^{k-1} - \sum_{j=k+1}^{m}
\Bigr) (-1)^{j} \iota_{[v_i,v_k] \wedge v_1 \cdots \widehat{v}_{i} \cdots \widehat{v}_{j}
  \cdots \widehat{v}_{k} \cdots \wedge v_m} \theta_{j}.
\end{split}
\end{equation}
The first term on the right-hand side of \eqref{I2_lemma_eq2} vanishes
if $\deg{x_{i}}>0$ for any $i$. The second term vanishes if more than
one $x_i$ has degree $>0$. Hence, the summation \eqref{I2_lemma_eq0}
of the terms \eqref{I2_lemma_eq2} is
\begin{equation} \label {I2_lemma_eq3}
\binom{m}{2}  \iota_{v_1 \wedge \cdots \wedge v_m} \omega +
\Bigl( - \sum_{i < k < j} + \sum_{i
      < j < k} - \sum_{j < i < k} \Bigr) (-1)^{i+k+j}
    \iota_{[v_i,v_k] \wedge v_1 \cdots \widehat{v}_{i} \cdots \widehat{v}_{j}
  \cdots \widehat{v}_{k} \cdots \wedge v_m} \theta_{j}
\end{equation}
Combining the above with \eqref{I2_lemma_eq1.5} completes the proof.

\end{proof}

\begin{lemmaapp}\label{I3_lemma}
For all $m \geq 3$ and  $x_{1}, \ldots,x_{m} \in L$, the following
equality holds:
\begin{equation} \label{I3_lemma_eq0}
\begin{split}
  I^{m \geq 3}_{(3)}(x_{1},\ldots,x_{m}) &= (-1)^{\binom{m+1}{2}} (-1)^{m}  \left(
    \binom{m}{2} -m +1 \right) \iota_{\vk{m}} \omega.
\end{split}
\end{equation}
\end{lemmaapp}
\begin{proof}
Let $\sigma \in \Sh(k,m-k)$. We have the following equalities:
\begin{equation} \label{I3_lemma_eq1}
f_{1}(l_{m}(x_{\sigma(1)},\ldots,x_{\sigma(m)})) = (-1)^{\binom{m+1}{2}}\iota_{\vsk{m}} \omega,
\end{equation}
and, for all $k < m$:
\begin{equation}\label{I3_lemma_eq2}
\begin{split}
\scalebox{0.86}{$f_{m+1 -k}(l_{k}(x_{\sigma(1)},\hdots,x_{\sigma(k)}) ,\ldots,x_{\sigma(m)})$}
&=\scalebox{0.86}{$ -(-1)^{\binom{m-k+2}{2}} S_{m+1-k} (l_{k}(x_{\sigma(1)},\hdots,x_{\sigma(k)}) ,\ldots,x_{\sigma(m)})$}\\
&\scalebox{0.86}{$=-(-1)^{\binom{m-k+2}{2}} (-1)^{1}\iota_{v_{\sigma(k+1)} \wedge \cdots \wedge
v_{\sigma(m)}} l_{k}(x_{\sigma(1)},\hdots,x_{\sigma(k)})$} \\
&\scalebox{0.86}{$=-(-1)^{\binom{k+1}{2}}(-1)^{\binom{m-k+2}{2}}\iota_{v_{\sigma(1)} \wedge \cdots \wedge
v_{\sigma(m)}} \omega.$}
\end{split}
\end{equation}
The second-to-last equality above follows from the fact that $\deg{l_k} >0$ for $k \geq 3$.
Combining \eqref{I3_lemma_eq1}, and \eqref{I3_lemma_eq2}, with the
definition of $I^{m}_{(3)}$ \eqref{I_maps} gives
\begin{equation}\label{I3_lemma_eq3}
\begin{split}
\scalebox{0.93}{$I^{m \geq 3}_{(3)}(x_{1},\ldots,x_{m}) = \displaystyle{\sum_{k=3}^{m} \sum_{\sigma
  \in  \Sh(k,m-k)}} \chi(\sigma)
(-1)^{k(m-k)}(-1)^{\binom{k+1}{2}}(-1)^{\binom{m-k+2}{2}}\iota_{v_{\sigma(1)}
\wedge \cdots \wedge v_{\sigma(m)}} \omega.$}
\end{split}
\end{equation}
The sum on the right-hand side above vanishes if, for any $i$,
$\deg{x_i}>0$. Non-zero summands above have
$\chi(\sigma)=(-1)^{\sigma}$, and since $\omega$ is skew-symmetric,
reordering the vector fields will cancel this sign.
The number of unshuffles appearing in the summation is $\binom{m}{k}$,
therefore, summing over $\sigma$ gives
\begin{equation}\label{I3_lemma_eq4}
\begin{split}
I^{m \geq 3}_{(3)}(x_{1},\ldots,x_{m}) = \sum_{k=3}^{m}
(-1)^{k(m-k)}(-1)^{\binom{k+1}{2}}(-1)^{\binom{m-k+2}{2}} \binom{m}{k}
\iota_{\vk{m}} \omega.
\end{split}
\end{equation}
It's easy to see that
$
(-1)^{k(m-k)} (-1)^{\binom{k+1}{2}}(-1)^{\binom{m-k+2}{2}}\! =\!\ksgn{m}(-1)^{m}(-1)^{k}.
$
Substituting the above sign into \eqref{I3_lemma_eq4} and using the
 fact that $\sum_{k=0}^{m} \binom{m}{k} (-1)^{k} =0$
gives the equality \eqref{I3_lemma_eq0}.

\end{proof}

\begin{lemmaapp}\label{J_lemma1}
For all $m \geq 2$ and  $x_{1}, \ldots,x_{m} \in L$ the following
equality holds:
\begin{equation}\label{Jlem_eq0}
\begin{split}
J^{ m\geq 2}(x_1,\hdots,x_m)  &= \sum_{i=1}^{m} (-1)^{i-1} \bbrac{f_{1}(x_{i})}{f_{m-1}(x_{1},\ldots,\widehat{x}_{i},\ldots,x_{m})}.
\end{split}
\end{equation}
\end{lemmaapp}
\begin{proof}
Recalling the definition of $J^m$ \eqref{J_maps}, it is easy to see that 
$J^{2}(x_1,x_2) \break = \bbrac{f_{1}(x_1)}{f_{1}(x_2)}$.
For the  $m >2$ case,  it follows from the definition of the bracket \eqref{brackets_app}
that 
\begin{equation} \label {Jlem_eq1}
\begin{split}
   J^{m \geq 3} (x_1,\hdots,x_m) &= \bbrac{f_{1}(x_{1})}{f_{m-1}(x_{2},\hdots,x_{m})}\\
 & \quad + \sum_{i \geq 2} \chi(\tau(i)) (-1)^{m}
 \bbrac{f_{m-1}(x_1, x_2,\ldots, \widehat{x}_i,\ldots,x_m)}{f_{1}(x_{i})}.
\end{split}
\end{equation}

Above $x_{i}=x_{\tau(m)}$, so
$\chi(\tau(i))=(-1)^{m-i}\epsilon(\tau(i))=(-1)^{m-i}(-1)^{\deg{x_i}\sum_{j>i} \deg{x_j}}$.
It follows from the antisymmetry of the bracket and
the definition of the structure maps that the summation on the
right-hand side
of \eqref {Jlem_eq1} is
\begin{equation}\label{Jlem_eq3}
\sum_{i \geq 2}  (-1)^{i-1} \bbrac{f_{1}(x_{i})}{f_{m-1}(x_1, x_2,\ldots, \widehat{x}_i,\ldots,x_m)}.
\end{equation}
Hence, the equality \eqref{Jlem_eq0} holds.

\end{proof}

\begin{lemmaapp}\label{J_lemma2}
For all $m \geq 3$ and  $x_{1}, \ldots,x_{m} \in L$ the following
equality holds:
\begin{equation}
\begin{split}
& \scalebox{0.84}{$\displaystyle{-(-1)^{\binom{m}{2}} J^{m \geq 3}(x_1,\ldots,x_m) = 2 \Bigl(\sum_{i < k < j}
  - \sum_{i < j<k} + \sum_{j < i < k} \Bigr) (-1)^{i+j+k}
  \iota_{[v_i,v_k] \wedge v_1 \wedge \cdots \widehat{v}_{i} \cdots \widehat{v}_{j} \cdots
\widehat{v}_{k} \cdots \wedge v_{m}} \theta_{j}}$} \\
& \scalebox{0.84}{$\quad \quad \displaystyle{+ \Bigl( \sum_{i <j} - \sum_{j<i} \Bigr) (-1)^{i+j}
\iota_{v_{1} \wedge \cdots \widehat{v}_{i} \cdots \widehat{v}_{j}
\cdots \wedge v_{m} } \L_{v_{i}} \theta_{j}
 - 2 \sum_{i < j} (-1)^{i+j}\iota_{[v_i,v_j] \wedge
v_{1} \wedge \cdots \widehat{v}_{i} \cdots \widehat{v}_{j}
\cdots \wedge v_{m} } \bar{A}(m-1)}$} \\
& \scalebox{0.84}{$\quad \quad \displaystyle{- \sum_{i} (-1)^{i} \iota_{v_1 \wedge
\cdots \widehat{v}_{i} \cdots \wedge v_{m}} \L_{v_{m}} \bar{A}(m-1)}.$}
\end{split}
\end{equation}
\end{lemmaapp}

\begin{proof}
Lemma \ref{J_lemma1} and the definitions of the bracket $\bbrac{\cdot}{\cdot}$
and $f_{m-1}$ imply that
\begin{equation} \label{Jlem_eq20}
\begin{split}
\scalebox{0.82}{$-(-1)^{\binom{m}{2}} J^{ m > 2}(x_1,\hdots,x_m)$} &\scalebox{0.82}{$= \displaystyle{\sum_{i} (-1)^{i-1}
\bigl( \L_{v_{i}} S_{m-1}(x_{1},\ldots,\widehat{x}_{i},\ldots,x_{m})
+ \L_{v_{i}} \iota_{v_1 \wedge \cdots \widehat{v}_{i} \cdots \wedge
v_{m}} \bar{A}(m-1) \bigr)}$}.
\end{split}
\end{equation}
The definition of $S_{m-1}$ \eqref{Sm_def} implies that
the first summation on the right-hand side of \eqref{Jlem_eq20} is
\begin{equation}\label{Jlem_eq21}
\begin{split}
\sum_{i} (-1)^{i-1}\L_{v_{i}}
S_{m-1}(x_{1},\ldots,\widehat{x}_{i},\ldots,x_{m}) &=
\sum_{i < j} (-1)^{i+j+1}\Bigl( \L_{v_{j}} \iota_{v_{1} \wedge \cdots
\widehat{v}_{i} \cdots \widehat{v}_{j} \cdots \wedge v_m} \theta_{i}
\\
& - \L_{v_{i}} \iota_{v_{1} \wedge \cdots
\widehat{v}_{i} \cdots \widehat{v}_{j} \cdots \wedge v_m} \theta_{j} \Bigr)
\end{split}
\end{equation}
The commutator \eqref{Cartan_commutator} implies that
$$\iota_{[v_j, v_{1} \wedge \cdots
\widehat{v}_{i} \cdots \widehat{v}_{j} \cdots \wedge v_m]} =
\L_{v_{j}} \iota_{v_{1} \wedge \cdots\widehat{v}_{i} \cdots \widehat{v}_{j} \cdots \wedge v_m}
- \iota_{v_{1} \wedge \cdots\widehat{v}_{i} \cdots \widehat{v}_{j}
\cdots \wedge v_m} \L_{v_{j}}.$$
This and the definition of the Schouten bracket \eqref{schouten} give the
following equalities:
 \begin{equation} \label{Jlem_eq23}
\begin{split}
&\scalebox{0.85}{$\displaystyle{\sum_{i < j} (-1)^{i+j+1}\L_{v_{j}} \iota_{v_{1} \wedge \cdots
\widehat{v}_{i} \cdots \widehat{v}_{j} \cdots \wedge v_m} \theta_{i}
=\Bigl( -\sum_{i < j < k} + 2\sum_{j < i < k} \Bigr) (-1)^{i+j+k}
\iota_{[v_i,v_k] \wedge v_1 \wedge \cdots \widehat{v}_{i} \cdots \widehat{v}_{j} \cdots
\widehat{v}_{k} \cdots \wedge v_{m}} \theta_{j}}$} \\
&\scalebox{0.85}{$\quad \quad  +\displaystyle{ \sum_{j < i} (-1)^{i+j+1} \iota_{v_{1} \wedge \cdots\widehat{v}_{i} \cdots \widehat{v}_{j}
\cdots \wedge v_m} \L_{v_{i}} \theta_{j}},$}
\end{split}
\end{equation}
\begin{equation}\label{Jlem_eq24}
\begin{split}
&\scalebox{0.85}{$\displaystyle{\sum_{i < j} (-1)^{i+j+1} \L_{v_{i}} \iota_{v_{1} \wedge \cdots
\widehat{v}_{i} \cdots \widehat{v}_{j} \cdots \wedge v_m} \theta_{j}
= \Bigl(-2\sum_{i < k < j} + \sum_{i < j<k} \Bigr) (-1)^{i+j+k}
\iota_{[v_i,v_k] \wedge v_1 \wedge \cdots \widehat{v}_{i} \cdots \widehat{v}_{j} \cdots \widehat{v}_{k} \cdots \wedge v_{m}} \theta_{j}}$} \\
& \scalebox{0.85}{$\quad \quad + \displaystyle{\sum_{i < j} (-1)^{i+j+1} \iota_{v_{1} \wedge \cdots\widehat{v}_{i} \cdots \widehat{v}_{j}
\cdots \wedge v_m} \L_{v_{i}} \theta_{j}}$}.
\end{split}
\end{equation}
As for the second summation on the right-hand side of
\eqref{Jlem_eq20}, note that the identity
 \eqref{Cartan_commutator} for the commutator gives
\begin{equation}\label {Jlem_eq25}
\begin{split}
\sum_{i} (-1)^{i-1}\L_{v_{i}} \iota_{v_1 \wedge \cdots \widehat{v}_{i} \cdots \wedge
v_{m}} \bar{A}(m-1) &=
\sum_{i} (-1)^{i-1} \bigl( \iota_{[v_i,v_1 \wedge \cdots \widehat{v}_{i} \cdots \wedge
v_{m}]}\\ &
\quad  +  \iota_{v_1 \wedge \cdots \widehat{v}_{i} \cdots \wedge
v_{m}} \L_{v_i} \bigr) \bar{A}(m-1).
\end{split}
\end{equation}
The definition of the Schouten bracket implies that
\begin{equation*}
\scalebox{0.97}{$\displaystyle{\sum_{i} (-1)^{i-1}  \iota_{[v_i,v_1 \wedge \cdots \widehat{v}_{i} \cdots \wedge
v_{m}]} \bar{A}(m-1)
 = -2\sum_{ i < j}
(-1)^{i+j} \iota_{[v_i,v_j] \wedge v_1 \wedge \cdots \widehat{v}_{i}
\cdots \widehat{v}_{j} \cdots\wedge v_{m}} \bar{A}(m-1).}$}
\end{equation*}
The above equality, along with \eqref{Jlem_eq25},
\eqref{Jlem_eq23}, and \eqref{Jlem_eq24},
gives the desired expression for $J^{m \geq 3}$.

\end{proof}

\begin{lemmaapp} \label{dtot_lem}
For all $m \geq 2$ and  $x_{1}, \ldots,x_{m} \in L$ the following
equality holds:
\begin{equation*}
\begin{split}
\scalebox{0.87}{$l'_1 f_{m}(x_1,\ldots,x_{m})$} &\scalebox{0.87}{$= \ksgn{m} \bigl(
   dS_{m}(x_{1},\ldots,x_{m}) + (-1)^{m} \iota_{v_1 \wedge \cdots \wedge v_{m}} \omega
   + \L_{\vk{m}} \bar{A}(m-1) \bigr).$}
\end{split}
\end{equation*}
\end{lemmaapp}

\begin{proof}
The definitions of $f_{m}$ and $l'_1$ imply that
\[
l'_1 f_{m}(x_1,\ldots,x_{m}) = \ksgn{m} \bigl( dS_m (x_1,\ldots,x_{m})
+ \dt \iota_{\vk{m}}\bar{A}(m) \bigr).
\]
The \v{C}ech differential commutes with interior product. Hence,
\[
\dt \iota_{\vk{m}}\bar{A}(m) = \iota_{\vk{m}} \delta \bar{A}(m) + d\iota_{\vk{m}} A^{n}\! +\! \sum_{i=1}^{n} (-1)^{mi+i} d \iota_{\vk{m}} A^{n-i}\!.
\]
Since $\bar{A}$ is a \v{C}ech-Deligne $n$-cocycle,
\[
\iota_{\vk{m}} \delta \bar{A}(m)  = -(-1)^{m}\iota_{\vk{m}} \sum_{i=1}^{n} (-1)^{(m-1)i}d A^{n-i}.\]
Hence,  Cartan's formula $\L_{\vk{m}}  = d\iota_{\vk{m}} -(-1)^{m} \iota_{\vk{m}}d$
implies that
\begin{equation*}
\begin{split}
\dt \iota_{\vk{m}}\bar{A}(m) &= d \iota_{v_1 \wedge \cdots v_{m}}
A^{n} + \sum_{i=1}^{n} (-1)^{(m-1)i} \L_{\vk{m}}  A^{n-i} \\
& = (-1)^{m}\iota_{\vk{m}}\omega + \sum_{i=0}^{n} (-1)^{(m-1)i} \L_{\vk{m}}  A^{n-i}.
\end{split}
\end{equation*}
The result then follows from the definition of
$\bar{A}(m-1)$ \eqref{Asign}.
\end{proof}

\subsection*{Proof of Theorem \ \ref{theorem.main-theorem}}
To prove that the maps $f_{k} \maps L^{\tensor
  k} \to L'$ give an $L_{\infty}$-morphism \cite[Def.\ 5.2]{Lada-Markl:1995},
we must verify that $\forall m \geq 1$
\begin{equation*}\label{bigger_eq}
\begin{split}
&\scalebox{0.82}{$l'_{1}f_{m}(x_1,\hdots,x_m) + \displaystyle{\sum_{j+k=m+1} \sum_{\sigma \in  \Sh(k,m-k) }}\chi(\sigma) (-1)^{k(j-1)+1} f_{j}(l_{k}(x_{\sigma(1)},\hdots,x_{\sigma(k)}),x_{\sigma(k+1)},\hdots,x_{\sigma(m)})$}\\
& \scalebox{0.82}{$\quad + \displaystyle{\sum_{s+t = m} \sum_{\substack{\tau \in \Sh(s,m-s) \\ \tau(1) <
    \tau(s+1)}}} \chi(\tau) (-1)^{s-1} (-1)^{(t-1) \sum_{p=1}^{s}
  \deg{x_{\tau(p)}}}
\bbrac{f_{s}(x_{\tau(1)},\hdots,x_{\tau(s)})}{f_{t}(x_{\tau(s+1)},\hdots,x_{\tau(m)})}
=0,$}
\end{split}
\end{equation*}
or, in our notation:
\begin{equation} \label{big_eq}
\bigl(l'_{1}f_{m} +  I^{m}_{(1)} + I^{m}_{(2)} + I^{m}_{(3)}
+ J^{m} \bigr) (x_1,\hdots,x_m) = 0 \quad \forall m
\geq 1.
\end{equation}

For $m=1$, \eqref{big_eq} holds, since $f_1$ is a chain map.
For $m=2$, we have $I^{2}_{(3)}=0$ by definition, and
it follows from Lemmas \ref{I1_lemma} and \ref{I2_lemma} that
\[
I^{2}_{(1)}(x_1,x_2) + I^{2}_{(2)}(x_1,x_2)  = - [v_1,v_2] -\iota_{v_2} l_1
\theta_1 + \iota_{v_1}l_{1}\theta_2 + \iota_{v_1 \wedge v_2}\omega -
\iota_{[v_1,v_2]}\bar{A}(1).
\]
From Lemma \ref{J_lemma1} we have
\[
J^{2}(x_1,x_2) = [v_1,v_2] - \iota_{v_1}d \theta_2 +
\iota_{v_2} d \theta_1 -dS_{2}(x_1,x_2) + \iota_{[v_1,v_2]} \bar{A}(1) -
\L_{v_{1} \wedge v_2} \bar{A}(1).
\]
Hence, the above equalities, along with Lemma \ref{dtot_lem}, imply that
the left-hand side of \eqref{big_eq} is $\iota_{v_1}
(l_1-d)\theta_{2} + \iota_{v_2} (d-l_1)\theta_{1} + 2\iota_{v_1 \wedge
  v_2}\omega$. If $\deg{x_1}=\deg{x_2}=0$, then $l_1=0$ and the $\theta_i$ are
Hamiltonian, i.e.,\ $-\iota_{v_1}d\theta_2 = \iota_{v_2}d\theta_1 = -\iota_{v_{1} \wedge
v_{2}} \omega$. If  $\deg{x_{i}} > 0$, then $v_{i}=0$ and $l_1 \theta_i = d\theta_i$.
Therefore, in either case, \eqref{big_eq} holds.

For the $m \geq 3$ case, note that Lemma \ref{tech_lemma}
combined with Cartan's formula 
for the Lie derivative
implies that for any $x_{1},\ldots,x_{m} \in L$:
\begin{equation*}
\begin{split}
 (-1)^{m} \L_{\vk{m}} \bar{A}(m-1) &=
\sum_{i < j} (-1)^{i+j}\iota_{[v_i,v_j] \wedge v_{1} \wedge \cdots
\widehat{v}_{i} \cdots \widehat{v}_{j} \cdots \wedge v_{m} } \bar{A}(m-1)\\
& \quad+ \sum_{i} (-1)^{i} \iota_{v_1 \wedge \cdots \widehat{v}_{i} \cdots
\wedge v_{m}} \L_{v_{m}} \bar{A}(m-1),
\end{split}
\end{equation*}
and
\begin{equation}\label{thm_eq2}
\begin{split}
\scalebox{0.82}{$(-1)^{m-1}\sum^{m}_{j=1} (-1)^{j} \L_{v_{1} \wedge \cdots \widehat{v}_{j}
  \cdots \wedge v_{m}} \theta_{j} $}&\scalebox{0.82}{$=
\Bigl(\displaystyle{\sum_{i < k < j}  - \sum_{i < j<k} + \sum_{j < i < k}} \Bigr) (-1)^{i+j+k}
  \iota_{[v_i,v_k] \wedge v_1 \wedge \cdots \widehat{v}_{i} \cdots \widehat{v}_{j} \cdots
\widehat{v}_{k} \cdots \wedge v_{m}} \theta_{j}$} \\
&  \scalebox{0.82}{$\quad+ \Bigl( \displaystyle{\sum_{i <j} - \sum_{j<i} }\Bigr) (-1)^{i+j}
\iota_{v_{1} \wedge \cdots \widehat{v}_{i} \cdots \widehat{v}_{j}
\cdots \wedge v_{m} } \L_{v_{i}} \theta_{j}$}
\end{split}
\end{equation}
Combining the above equalities with  Lemmas \ref{I2_lemma},
\ref{I3_lemma}, and \ref{J_lemma2} gives
\begin{equation}\label{thm_eq3}
\begin{split}
(I^{m}_{(2)} + I^{m}_{(3)} +  J^{m})(x_{1},\ldots,x_{m}) &= -(-1)^{\binom{m}{2}}
\Bigl ( (-1)^{m-1}\sum^{m}_{j=1} (-1)^{j} \L_{v_{1} \wedge \cdots \widehat{v}_{j}
  \cdots \wedge v_{m}} \theta_{j}  \\
& \quad  -  (-1)^{m} \L_{\vk{m}} \bar{A}(m-1) + (m-1)
  \iota_{\vk{m}}\omega \Bigr ).
\end{split}
\end{equation}
Cartan's formula also implies that
\[
\sum^{m}_{j=1} (-1)^{j} \L_{v_{1} \wedge \cdots \widehat{v}_{j}
  \cdots \wedge v_{m}} \theta_{j} = dS_{m}(x_1,\ldots,x_{m}) -(-1)^{m-1}\sum_{j=1}^{m} (-1)^{j} \iota_{v_1
 \wedge \cdots \widehat{v}_{j} \cdots \wedge v_m} d\theta_{j}.
\]
Using this, along with Eq.\ (\ref{thm_eq3}) and
the results of Lemmas \ref{I1_lemma} and \ref{dtot_lem}, we conclude
that the left-hand side of \eqref{big_eq} is
\begin{equation}\label{thm_eq5}
 \scalebox{0.92}{$\displaystyle{-(-1)^{\binom{m}{2}}  \Bigl( \sum_{i=1}^{m}
(-1)^{i}\iota_{v_{1} \wedge \cdots \wedge \widehat{v}_{i} \wedge \cdots \wedge
v_{m}} (l_{1} \theta_{i})
 -\sum_{j=1}^{m} (-1)^{j} \iota_{v_1\wedge \cdots
\widehat{v}_{j} \cdots \wedge v_m} d\theta_{j} +   m\iota_{\vk{m}}\omega \Bigr )}.$}
\end{equation}
If all $x_{i}$ are degree 0, then $l_1=0$, and all $\theta_{i}$ are
Hamiltonian, which implies that  
\[
\sum_{j=1}^{m} (-1)^{j} \iota_{v_1\wedge \cdots
\widehat{v}_{j} \cdots \wedge v_m} d\theta_{j} =m \iota_{\vk{m}}\omega.
\]
If  $\deg{x_{k}} >0$ for some $x_k$ then $v_k=0$ and $l_1
\theta_k =d\theta_k$. Hence, in either case, \eqref{big_eq} holds.
This completes the proof of the theorem. \qed

\section{A recognition principle for homotopy fibers of \texorpdfstring{$L_\infty$}{Loo}-morphisms}

 In this section we provide a proof of the recognition principle for homotopy fibers
 of $L_\infty$-algebra morphisms that has been used in Section \ref{LooXomega}.  The proof is based on the following two facts recalled in the Introduction. First,  every $L_\infty$-morphism $f_\infty:\mathfrak{g}\to A$ to a dg-Lie algebra $A$
uniquely factors as
$
\mathfrak{g}\xrightarrow{v_{\mathfrak{g}}} \mathcal{R}(\mathfrak{g})\xrightarrow{\xi_A\circ \mathcal{R}(f_\infty)}A,
$
 where $\xi_A:\mathcal{R}(A)\to A$ is the dg-Lie algebra morphism in the factorization of the identity of $A$ as
$
A\xrightarrow{v_{A}} \mathcal{R}(A)\xrightarrow{\xi_A}A$.
Second,  the adjunction
   $(\mathcal{R} \dashv i)$ 
   induces an equivalence between the homotopy theories of dg-Lie algebras
   and $L_\infty$-algebras and so, if
$f_\infty:\mathfrak{g}\to \mathfrak{h}$ is an $L_\infty$-morphism between two $L_\infty$-algebras,
then an $L_\infty$-algebra $\mathfrak{k}$ presents the homotopy fiber of $f_\infty$ if $\mathfrak{k}$ is $L_\infty$-quasi-isomorphic to the homotopy fiber of $\mathcal{R}(f_\infty):\mathcal{R}(\mathfrak{g})\to \mathcal{R}(\mathfrak{h})$ in the category of dglas.

\begin{lemma}\label{lemma-zero}
Let $\mathfrak{g}$ be an $L_\infty$-algebra, $A$ a dgla, and  $f_\infty: \mathfrak{g} \to A$ an $L_\infty$ morphism. Let $p_A:B \to A$ be a fibration in the category of dglas, with $H_\bullet(B)=0$. The fiber product
\[
\scalebox{0.92}{
\xymatrix{
\mathcal{R}(\mathfrak{g})\times_A B\ar[r]^-{\pi_{B}}\ar[d]_{\pi_{\mathcal{R}(\mathfrak{g})}} & B\ar[d]^{p_A}\\
\mathcal{R}(\mathfrak{g})\ar[r]^{\xi_A\circ \mathcal{R}(f_\infty)}&A
}
}
\]
is a dgla model for the homotopy fiber of $f_\infty$.
\end{lemma}
\begin{proof}
Consider the commutative diagram of dglas
\[
\scalebox{0.92}{
\xymatrix{
\mathcal{R}(\mathfrak{g})\times_A B\ar[rr]^{{(\mathcal{R}(f_\infty),\mathrm{id}_B)}}\ar[d]_{\pi_{\mathcal{R}(\mathfrak{g})}} &&\mathcal{R}(A)\times_A B\ar[r]^-{\pi_{B}}\ar[d]^{\pi_{\mathcal{R}(A)}} & B\ar[d]^{p_A} \\
\mathcal{R}(\mathfrak{g})\ar[rr]^{\mathcal{R}(f_\infty)}&&\mathcal{R}(A)\ar[r]^{\xi_A}&A
}
}
\]
where the rightmost diagram and the outer diagram are pullbacks. By
the pasting law, also the leftmost diagram is a pullback.
Since $p_A$ is a fibration and $\xi_{A}$ is a weak equivalence,
the map $\pi_{\mathcal{R}(A)}$ is a fibration and the map $\pi_{B}$ is a weak
equivalence. It follows that
$\pi_{\mathcal{R}(A)}$ is a fibrant replacement of $0\to \mathcal{R}(A)$. Hence $\mathcal{R}(\mathfrak{g})\times_A B$ is a model for the homotopy fiber of $\mathcal{R}(f_\infty)$ in the category of dglas.
\end{proof}
\begin{theorem}
Let $\mathfrak{g}$ be an $L_\infty$-algebra, $A$ a dgla, and  $f_\infty: \mathfrak{g} \to A$ an $L_\infty$ morphism. Let $p_A:B \to A$ be a fibration  in the category of dglas, with $H_\bullet(B)=0$. Assume we have
a commutative diagram of $L_\infty$-algebras
\[\scalebox{0.92}{
\xymatrix{
(\mathfrak{g}\times_A B, Q)\ar[r]^-{\pi_{B,\infty}}\ar[d]_{\pi_{\mathfrak{g},\infty}} & B\ar[d]^{p_A}\\
\mathfrak{g}\ar[r]^{f_\infty}&A
}
}
\]
 for a suitable $L_\infty$-structure $Q$ on the fiber product of chain complexes $\mathfrak{g}\times_A B$  of $p_A$ with the linear component of $f_\infty$, with  $\pi_{\mathfrak{g},\infty}$ and $\pi_{B,\infty}$ $L_\infty$-morphisms lifting the linear projections $\pi_{\mathfrak{g}}$ and $\pi_B$. Then
 $(\mathfrak{g}\times_A B, Q)$ is a model for the homotopy fiber of $f_\infty$.
 \label{DomenicosLemma}
\end{theorem}
\begin{proof}
Applying the rectification functor to the diagram of $L_\infty$-morphisms above we get a commutative diagram of dglas
\[
\scalebox{0.92}{
\xymatrix{
\mathcal{R}(\mathfrak{g}\times_A B,Q)\ar[r]^-{\mathcal{R}(\pi_{B,\infty})}\ar[d]_{\mathcal{R}(\pi_{\mathfrak{g},\infty})} & \mathcal{R}(B)\ar[d]^{\mathcal{R}(p_A)}\\
\mathcal{R}(\mathfrak{g})\ar[r]^{\mathcal{R}(f_\infty)}&\mathcal{R}(A)
}
}
\]
Using the counit of the adjunction $(\mathcal{R} \dashv i)$, we can extend this to a commutative diagram of dglas
\[
\scalebox{0.92}{
\xymatrix{
\mathcal{R}(\mathfrak{g}\times_A B,Q)\ar[r]^-{\mathcal{R}(\pi_{B,\infty})}\ar[d]_{\mathcal{R}(\pi_{\mathfrak{g},\infty})} & \mathcal{R}(B)\ar[d]^{\mathcal{R}(p_A)}\ar[r]^{\xi_B}&B\ar[d]^{p_A}\\
\mathcal{R}(\mathfrak{g})\ar[r]^{\mathcal{R}(f_\infty)}&\mathcal{R}(A)\ar[r]^{\xi_A}&A
}
}
\]
By the universal property of the pullback of dglas, the outer rectangle
 is equivalent to the datum of a morphism of dglas $
\psi: \mathcal{R}(\mathfrak{g}\times_A B,Q)\to \mathcal{R}(\mathfrak{g})\times_A B$,
where $\mathcal{R}(\mathfrak{g})\times_A B$ is the pullback of dglas
\[
\scalebox{0.92}{
\xymatrix{
\mathcal{R}(\mathfrak{g})\times_A B\ar[d]_{\pi_{\mathcal{R}({\mathfrak{g}})}}\ar[rr]^-{\pi_B}& &B\ar[d]^{p_A}\\
\mathcal{R}(\mathfrak{g})\ar[rr]^{\xi_A\circ \mathcal{R}(f_\infty)}&&A
}
}
\]
The morphism $\psi$ will satisfy $
\pi_B\circ\psi=\xi_B\circ \mathcal{R}(\pi_{B,\infty})$ and $
 \pi_{\mathcal{R}(\mathfrak{g})}\circ\psi=\mathcal{R}(\pi_{\mathfrak{g},\infty})$.
By Lemma \ref{lemma-zero}, the dgla $\mathcal{R}(\mathfrak{g})\times_A B$  is a dgla model for the homotopy fiber of $f_\infty$. Then to conclude we only need to show that $\psi$ is a quasi-isomorphism. This is equivalent to proving that the $L_\infty$-morphism $\eta$ given by the composition
\[
 \eta :
(\mathfrak{g}\times_A B,Q)\xrightarrow{v_{(\mathfrak{g}\times_A B,Q)}} \mathcal{R}(\mathfrak{g}\times_A B,Q)\xrightarrow{\psi} \mathcal{R}(\mathfrak{g})\times_A B
\]
is a quasi-isomorphism. The linear part $\eta_1$ of $\eta$ is determined by its compositions with the linear projections to $\mathcal{R}(\mathfrak{g})$ and to $B$. We have
$
\pi_B\circ \eta_1
=(\pi_B\circ \eta)_1=(\pi_B\circ\psi\circ v_{(\mathfrak{g}\times_A B,Q)})_1=
(\xi_B\circ\mathcal{R}(\pi_{B,\infty})\circ v_{(\mathfrak{g}\times_A B,Q)})_1
=(\pi_{B,\infty})_1=\pi_B
$ 
and, similarly,
$
\pi_{\mathcal{R}(\mathfrak{g})}\circ \eta_1
=(\pi_{\mathcal{R}(\mathfrak{g})}\!\circ \eta)_1=(\pi_{\mathcal{R}(\mathfrak{g})}\!\circ\psi\!\circ v_{(\mathfrak{g}\times_A B,Q)})_1=
(\mathcal{R}(\pi_{\mathfrak{g},\infty})\!\circ\! v_{(\mathfrak{g}\times_A B,Q)})_1
\! =(v_{\mathfrak{g}}\!\circ\!\pi_{\mathfrak{g},\infty})_1= \break(v_{\mathfrak{g}})_1\!\circ\pi_{\mathfrak{g}}
$. 
This means that the map of chain complexes
$\eta_1: \mathfrak{g}\times_A B\to  \mathcal{R}(\mathfrak{g})\times_A B$
is given by
$\eta_1=((v_{\mathfrak{g}})_1,\mathrm{id}_B)$.
Now consider the commutative diagram
\[
\scalebox{0.92}{
\xymatrix{
\mathfrak{g}\times_A B\ar[rr]^{((v_{\mathfrak{g}})_1,\mathrm{id}_B)}\ar[d]_{\pi_{\mathfrak{g}}} &&\mathcal{R}(\mathfrak{g})\times_A B\ar[rr]^-{\pi_{B}}\ar[d]_{\pi_{\mathcal{R}(\mathfrak{g})}} && B\ar[d]^{p_A}\\
\mathfrak{g}\ar[rr]^{(v_{\mathfrak{g}})_1}&&\mathcal{R}(\mathfrak{g})\ar[rr]^{\xi_A\circ \mathcal{R}(f_\infty)}&&A
}
}
\]
The rightmost subdiagram is a pullback by definition, while the total diagram is
 \[
 \scalebox{0.92}{
\xymatrix{
\mathfrak{g}\times_A B\ar[r]^{\pi_B}\ar[d]_{\pi_{\mathfrak{g}}} & B\ar[d]^{p_A}\\
\mathfrak{g}\ar[r]^{(f_\infty)_1}&A
},
}
\]
since
$
\xi_A\circ \mathcal{R}(f_\infty)\circ (v_{\mathfrak{g}})_1=(\xi_A\circ \mathcal{R}(f_\infty)\circ v_{\mathfrak{g}})_1=(f_\infty)_1,$
and so it is a pullback by hypothesis. Then, by the pasting law, also
the leftmost subdiagram is a pullback.
The map $\pi_{\mathcal{R}(\gg)}$ is a fibration, since $p_{A}$ is fibration,
and all chain complexes are fibrant. Hence,
since $(v_{\mathfrak{g}})_1$ is a quasi-isomorphism, its pullback
$\eta_1=((v_{\mathfrak{g}})_1,\mathrm{id}_B)$ is also a quasi-isomorphism.
\end{proof}


\begin{thebibliography}{10}

\bibitem{HDA6}
J.\ Baez and A.\ Crans, {\it Higher-dimensional algebra VI: Lie 2-algebras},
Theory Appl. Categ. {\bf 12} (2004), 492--528.
\href{http://arxiv.org/abs/math/0307263}
{\tt arXiv:math/0307263v6}.

\bibitem{Baez-Dolan}
J.\ Baez and  J.\ Dolan, {\it Higher-dimensional algebra and topological
quantum field theory}, {J.\ Math.\ Phys.} {\bf 36} (1995), 6073--6105. 
\href{http://arxiv.org/abs/q-alg/9503002}{\tt arXiv:q-alg/9503002}.

\bibitem{Baez:2008bu}
J.\ Baez, A.\ Hoffnung, and C.\ Rogers, {\it Categorified symplectic
geometry and the classical string}, { Comm.\ Math.\ Phys.} \textbf{293} (2010), 701--715.
\href{http://arxiv.org/abs/0808.0246}
{\tt arXiv:0808.0246}.

\bibitem{RogersString}
J.~Baez, C.~Rogers,
\newblock {\it Categorified symplectic geometry and the string Lie 2-algebra},
\newblock Homology Homotopy Appl. 12 (2010), 221-236
\newblock  \href{http://arxiv.org/abs/0901.4721}
{\tt arXiv:0901.4721}

\bibitem{Bongers13}
S. Bongers,
\newblock {\it Geometric quantization of symplectic and Poisson manifolds}
\newblock MSc thesis, Utrecht, January 2014,
\newblock  \href{http://ncatlab.org/schreiber/show/master+thesis+Bongers}
{\tt ncatlab.org/schreiber/\allowbreak show/master+thesis+Bongers}

\bibitem{Bott-Tu}
R. Bott\ and\ L. W. Tu, {\it Differential forms in algebraic topology}, Graduate Texts in Mathematics, Vol. 82, Springer, New York, 1982.


\bibitem{Brown}
K.~Brown,
\newblock {\it Abstract homotopy theory and generalized sheaf cohomology}, Trans. Amer. Math. Soc. {\bf 186} (1973), 419-458.

\bibitem{Brylinski:1993}
J.-L.~Brylinski, {\sl Loop Spaces, Characteristic Classes and Geometric
Quantization}, Birkhauser, Boston, 1993.

\bibitem{SKT}
A.\ Canas da Silva, Y.\ Karshon, S.\ Tolman,
{\it Quantization of presymplectic manifolds and circle actions},
Trans. Amer. Math. Soc. 352 (2000), 525-552.
 \href{http://arxiv.org/abs/dg-ga/9705008}
 {\tt arXiv:dg-ga/9705008}.

\bibitem{CdS-Weinstein:1999}
A.\ Cannas da Silva and A.\ Weinstein,
\textsl{Geometric Models for Noncommutative Algebras},
Berkeley Mathematics Lecture Notes \textbf{10}, Amer.\ Math.\ Soc.,
Providence, 1999.

\bibitem{Collier:2011}
B.\ Collier, {\it Infinitesimal symmetries of Dixmier-Douady gerbes}.
\href{http://arxiv.org/abs/1108.1525}
{\tt arXiv:\allowbreak 1108.1525}.

\bibitem{companionArticle}
D.~Fiorenza, C.~Rogers, U.~Schreiber,
\newblock {\it  Higher geometric prequantum theory},
 \href{http://arxiv.org/abs/1304.0236}
 {\tt arXiv:1304.0236}

\bibitem{FScSt}
D.~Fiorenza, U.~Schreiber, J.~Stasheff
\newblock {\it {\v C}ech cocycles for differential characteristic classes},
\newblock { Adv. Theor. Math. Phys.} {\bf 16} (2012) 149-250
\newblock \href{http://arxiv.org/abs/1011.4735}
{\tt arXiv:1011.4735}

\bibitem{Forger} M.\ Forger, C.\ Paufler, and H.\ R\"{o}mer,
  {\it The Poisson bracket for Poisson forms in multisymplectic field
  theory}, {Rev.\ Math.\ Phys.} \textbf{15} (2003), 705--744.
 \href{http://arxiv.org/abs/math-ph/0202043}
 {\tt math-ph/0202043}.

\bibitem{Freed}
D.~Freed,
\newblock {\it Higher algebraic structures and quantization},
\newblock Commun. Math. Phys. {\bf 159} (1994) 343-398,
\newblock \href{http://arxiv.org/abs/hep-th/9212115}
{\tt hep-th/9212115}


\bibitem{FRZ:2013}
Y.\ Fr\'{e}gier, C.\ Rogers, and M.\ Zambon, {\it Homotopy moment
maps}.  \href{http://arxiv.org/abs/1304.2051}
{\tt arXiv:1304.2051}


\bibitem{Hinich}
V.~Hinich, {\it Homological algebra of homotopical algebras},
\newblock{Comm. in Algebra} {\bf 25} (1997), 3291-3323.

\bibitem{Hinich:2001}
V.\ Hinich, {\it DG coalgebras as formal stacks}, { J.\ Pure Appl.\
  Algebra} \textbf{162} (2001), 209-250.
   \href{http://arxiv.org/abs/math/9812034}
   {\tt arXiv:math/981203}.

\bibitem{Huebschmann}
J.~Huebschmann, {\it The Lie algebra perturbation lemma}, {Higher structures in geometry and physics}, Progr. Math., \textbf{287}, 159-179, Birkh\"auser/Springer, New York, 2011, {\tt arXiv:0708.3977}.


\bibitem{Kanatch}
I.V. Kanatchikov,
\newblock{\it Geometric (pre)quantization in the polysymplectic approach to field theory}, 
{in \it  Differential geometry and its applications} (Opava, 2001), 309Ð321, Math. Publ., 3, Silesian Univ. Opava, Opava, 2001. 
\newblock \href{http://arxiv.org/abs/hep-th/0112263}
{\tt hep-th/0112263}.


\bibitem{Kostant:1970}
B.\ Kostant, {\it Quantization and unitary representations}, {Lecture Notes in Math.}
\textbf{170} (1970), 87--208.

\bibitem{Lada-Markl:1995} T.\ Lada and M.\ Markl, {\it Strongly homotopy Lie
  algebras}, {Comm.\ Algebra.} \textbf{23} (1995),
  2147--2161.
  \href{http://arxiv.org/abs/hep-th/9406095}
  {\tt arXiv:hep-th/9406095}.

\bibitem{Loday-Vallette} J.-L.\ Loday and B.\ Vallette,
  \textsl{Algebraic operads}, Grundlehren Math.\ Wiss.\ {\bf 346},
  Springer, Heidelberg, 2012.

\bibitem{Lurie}
J.~Lurie,
\newblock {\it On the classification of topological field theories},
\newblock Current Developments in Mathematics Volume 2008 (2009), 129-280,
\newblock \href{http://arxiv.org/abs/0905.0465}
{\tt arXiv:0905.0465}


\bibitem{NikolausSchreiberStevensonI}
T.~Nikolaus, U.~Schreiber, D.~Stevenson,
\newblock {\it Principal $\infty$-bundles, I: General theory},
\newblock to appear in J. of Homotopy and Related Structures (2014)
 \newblock \href{http://arxiv.org/abs/1207.0248}
 {\tt arXiv:1207.0248}

\bibitem{NikolausSchreiberStevensonII}
T.~Nikolaus, U.~Schreiber, and D.~Stevenson,
\newblock {\it Principal $\infty$-bundles, II: Presentations},
\newblock to appear in J. of Homotopy and Related Structures (2014)
\newblock  \href{http://arxiv.org/abs/1207.0249}
{\tt arXiv:1207.0249}

\bibitem{Nuiten}
J.~Nuiten,
\newblock {\it Cohomological quantization of local prequantum boundary field theory},
\newblock MSc thesis, Utrecht, August 2013,
\newblock \href{http://ncatlab.org/schreiber/show/master+thesis+Nuiten}
{\tt ncatlab.org/schreiber/show/\allowbreak master+thesis+Nuiten}


\bibitem{Pridham}
J.~Pridham,
\newblock {\it Unifying derived deformation theories},
\newblock Adv. Math. 224 (2010), no.3, 772-826.
\newblock  \href{http://arxiv.org/abs/0705.0344}
{\tt arXiv:0705.0344}


\bibitem{Quillen:1969}
D.\ Quillen, {\it Rational homotopy theory}, {Ann. of Math. (2)} \textbf{90} (1969), 205--295.

\bibitem{Richter}
M. Richter
\newblock {\it A Lie infinity algebra of Hamiltonian forms in n-plectic geometry},
\newblock \href{http://arxiv.org/abs/1212.4596}
{\tt arXiv:1212.4596}


\bibitem{RogersThesis}
C.\ Rogers, \textit{Higher Symplectic Geometry}. Ph.D. thesis. Department of Mathematics, University of California, Riverside, 2011.  \href{http://arxiv.org/abs/1106.4068}
{\tt arXiv:1106.4068}.


\bibitem{Rogers:2010nw} C.\ Rogers,
  {\it $L_{\infty}$-algebras from
  multisymplectic geometry}, {Lett.\ Math.\ Phys.} \textbf{100} (2012), 29--50.
  \href{http://arxiv.org/abs/1005.2230}
  {\tt arXiv:1005.2230}.

\bibitem{rogers.2-plectic}
C.\ Rogers,
 {\it 2-plectic geometry, {C}ourant algebroids, and categorified prequantizaion}.
{J.\ Symplectic Geom.} {\bf 11} (2013), 53--91.
\newblock
\href{http://arxiv.org/abs/1009.2975}
{\tt arXiv:1009.2975}.

\bibitem{Roytenberg_L2A} D.\ Roytenberg,
 {\it On weak Lie 2-algebras},
in: P.\ Kielanowski et al. (eds.)\ XXVI Workshop on Geometrical Methods in Physics.
AIP Conference Proceedings \textbf{956}, pp. 180-198.
American Institute of Physics, Melville (2007).
 \href{http://arxiv.org/abs/0712.3461}
 {\tt arXiv:0712.3461}.

\bibitem{Roytenberg-Weinstein}D.\ Roytenberg and A.\ Weinstein,
  {\it Courant algebroids and strongly homotopy Lie algebras},
{Lett.\ Math.\ Phys.}  \textbf{46} (1998), 81--93.
 \href{http://arxiv.org/abs/math/9802118}
 {\tt arXiv:math/9802118}.

\bibitem{dcct}
U.~Schreiber, {\it Differential cohomology in a cohesive $\infty$-topos},
 \href{http://arxiv.org/abs/1310.7930}
 {\tt arXiv:1310.7930}

\bibitem{GQ}
J.-M.~Souriau,
\newblock {\it Structure of dynamical systems},
\newblock Progress in Mathematics, {\bf 149}, Birkh{\"a}user Boston, Boston, MA, 1997

\bibitem{Souriau:1967}
J.-M.~ Souriau, {\it Quantification g\'{e}om\'{e}trique: Applications},
{Ann.\ Inst.\ H.\ Poincar\'{e} Sect.\ A (N.S.)} \textbf{6}
(1967), 311--341.



\end{thebibliography}
\end{document}